\documentclass[conference,letterpaper]{sig-alternate-10pt}
\usepackage[
  pass,% keep layout unchanged 
]{geometry}% \documentclass[journal,letterpaper]{IEEEtran}
\usepackage{graphicx}
\usepackage{epstopdf}
\usepackage{amssymb,amsmath}
\usepackage{url}
\usepackage{cite}
\usepackage{color}
\usepackage{algorithm}
\usepackage{algpseudocode}

\DeclareMathOperator*{\argmin}{arg\,min}

\usepackage[caption=false,font=footnotesize]{subfig}

\def\sharedaffiliation{\end{tabular}\newline\begin{tabular}{c}}

\newtheorem{theorem}{Theorem}
\newtheorem{lemma}{Lemma}

\newtheorem{assumption}{Assumption}

\def\1{\mathbf{1}}
\newcommand{\vv}[1]{\boldsymbol{#1}}

\begin{document}

% Title Page
\title{Low Delay Random Linear Coding and Scheduling\\ Over Multiple Interfaces}

\numberofauthors{1}
\author{
  \alignauthor \ Andres Garcia-Saavedra, Mohammad Karzand, Douglas J. Leith\\
   \email{\{andres.garcia.saavedra, karzandm, doug.leith\}@scss.tcd.ie}
\sharedaffiliation
\begin{tabular}{ccc}
	\affaddr{School of Computer Science and Statistics}  \\
	\affaddr{Trinity College Dublin, Ireland}     
\end{tabular}
}

\maketitle

\begin{abstract}

Multipath transport protocols like MPTCP transfer data across multiple routes in parallel and deliver it in order at the receiver. When the delay on one or more of the paths is variable, as is commonly the case, out of order arrivals are frequent and head of line blocking leads to high latency. This is exacerbated when packet loss, which is also common with wireless links, is tackled using ARQ. This paper introduces Stochastic Earliest Delivery Path First (S-EDPF), a resilient low delay packet scheduler for multipath transport protocols.  S-EDPF takes explicit account of the stochastic nature of paths and uses this to minimise in-order delivery delay.  S-EDPF also takes account of FEC, jointly scheduling transmission of information and coded packets and in this way allows lossy links to reduce delay and improve resiliency, rather than degrading performance as usually occurs with existing multipath systems.  We implement S-EDPF as a multi-platform application that does not require administration privileges nor modifications to the operating system and has negligible impact on energy consumption.  We present a thorough experimental evaluation in both controlled environments and \emph{into the wild}, revealing dramatic gains in delay performance compared to existing approaches.

\end{abstract}

\section{Introduction}
\label{sec:intro}

% Current mobile communication devices embed multiple communication interfaces  to access the Internet (e.g. HSPA, LTE and IEEE 802.11).
% Transporting data between a source and destination \emph{in parallel} along multiple paths is well-recognised, at least in principle, as an effective means to improve  performance, e.g.  increase throughput~\cite{1709949, sundararajan2009network, fountainmptcp, khalili2013mptcp }, resilience (if one path  breaks, the connection can gracefully failover to the remaining paths) \cite{Zhang04atransport}, and load balancing~\cite{raiciu2011improving, wischik2008resource, wischik2011design}.

Current mobile communication devices embed multiple communication interfaces  to access the Internet (e.g. HSPA, LTE and IEEE 802.11).
Transporting data between a source and destination in parallel along multiple paths is well-recognised, at least in principle, as an effective means to improve  performance, e.g.  increase throughput~\cite{sundararajan2009network, fountainmptcp, khalili2013mptcp }, resilience (if one path  breaks, the connection can gracefully failover to the remaining paths) \cite{Zhang04atransport}, and load balancing~\cite{ wischik2011design}.
A well-known example of a multipath transport protocol is Multipath TCP (MPTCP)~\cite{ford2011tcp}, which extends the traditional single-path TCP (SPTCP) to stripe the data of a single connection across multiple routes or \emph{subpaths}. 
% Linux' MPTCP implements two congestion control algorithms: Linked Increase Algorithm (LIA) and Opportunistic LIA (OLIA)~\cite{khalili2013mptcp}. With LIA, an ACK received from sublow $k$ increases its congestion window $cw_k$ by $\min\{\frac{\max_i cw_i/rtt_i^2}{(\sum_i cw_i/rtt_i)^2}, \frac{1}{cw_k}\}$, and a packet loss on subflow $k$ decreases $cw_k$ by $cw_k/2$. This congestion control algorithm forces a compromise performance between responsiveness and congestion balancing. OLIA modifies the additive part as follows: Subplow $k$ increases $cw_k$ by $\frac{cw_k/rtt_k^2}{(\sum_i cw_i/rtt_i)^2} + \frac{\alpha_k}{cw_k}$, where $\alpha_k$ is set to optimise MPTCP's responsiveness to changes in current windows (see \cite{khalili2013mptcp}).
The most popular Linux implementation of MPTCP supports two schedulers to assign packets into subpaths~\cite{paasch2014experimental}: a \emph{round-robin} (RR) scheduler which iterates over each subflow regardless of their latency properties, and the \emph{lowest RTT First} (LowRTT) scheduler that gives priority to paths with lower round-trip times (RTT). 
% Although each subflow uses independent sequence numbers and congestion windows, the  default algorithm ruling these,  is not fully disjoint in MPTCP: an ACK received from sublow $k$ updates its congestion window $cw_k$ by $\min\{\frac{\max_i cw_i/rtt_i^2}{(\sum_i cw_i/rtt_i)^2}, \frac{1}{cw_k}\}$, and a packet loss on subflow $k$ decreases $cw_k$ by $cw_k/2$. This congestion control algorithm represents a compromise solution between responsiveness and congestion balancing.  
% The Linux kernel implementation of MPTCP supports two schedulers~\cite{paasch2014experimental}: a \emph{round-robin} (RR) scheduler which iterates over each subflow independently of their individual characteristics, and the \emph{lowest-RTT-First} (LowRTT) scheduler that gives priority to subflows with lower round-trip times (RTT). 
Finally, similarly to SPTCP, packet loss recovery is achieved by means of an ARQ mechanism, i.e., feedback from the receiver is used to detect and retransmit lost packets.

\subsection{The head-of-line blocking problem}

\begin{figure}[t!]
\centering
\includegraphics[width=0.93\columnwidth ]{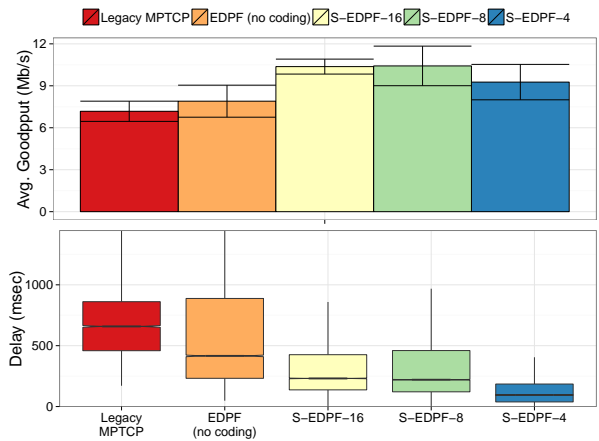}
\vspace{-5mm}
\caption{Multipath uploads (2.4-Ghz WiFi + LTE) in a home environment. Mean goodput and standard error at the top. Box and whiskers for packet  delay at the bottom. Details in \S\ref{sec:performance}.}
\label{fig:exp:mptcp}
\vspace{-5mm}
\end{figure}

Although  a promising technology, building an efficient, practically usable, multipath transfer mechanism remains highly challenging. Indeed, issues that did not exist in the single-path context appear now in the multipath paradigm, fostering a rich amount of research on packet scheduling~\cite{edpf, paasch2014experimental}, loss recovery~\cite{fountainmptcp, sundararajan2009network}, and rate and congestion control~\cite{wischik2011design}. In this paper we design and prototype S-EDPF (Stochastic Earliest Delivery Path First), a low-delay packet scheduler for multipath transport protocols. Like MPTCP, we buffer packets at the receiver as they arrive until they can be delivered to the application \emph{in order}. This causes an additional  delay to the delivery of packets  as these may have to await others that  ($i$) arrive out of order (a frequent event with multipath transporting), or ($ii$) have been lost and need to be recovered (e.g. by retransmission). This is known as \emph{head-of-line blocking} (HOL) and its effects on the performance of MPTCP are illustrated by the experiments of Fig.~\ref{fig:exp:mptcp}. These consist of a set of uplink transmissions from a laptop attached to a home-based WiFi network (in the 2.4Ghz band) and a 3G/4G dongle (with Meteor, an Irish  provider), and show how delay scales up to half a second with MPTCP  in a real environment, an intolerable value for real-time applications. 

% This poor performance can be explained by the next two phenomena.

% \footnote{See \S\ref{sec:performance} for the details of our experimental campaign.}
% \footnote{ \S\ref{sec:performance} details our experimental campaign.}

% \subsubsection{Out of order arrivals}
{\bf Out of order arrivals.} An adequate scheduling of packets is of paramount importance to minimise reordering delay~\cite{paasch2014experimental}. Besides the two aforementioned schedulers used in the Linux implementation of MPTCP, it is worth highlighting EDPF (Earliest Delivery Path First) \cite{edpf}, which assigns  packets to the subpath with earliest expected arrival. If links are deterministic and there are no losses, EDPF is optimal. However, transmission rates and propagation delays \emph{are random in nature}~\cite{paasch2014experimental, doug-lte-measurements}, and taking decisions based only on averages renders suboptimal performance as Fig.~\ref{fig:example-stochastic} illustrates.\footnote{The study in \cite{chen2013measurement} with MPTCP and US mobile ISPs show that 20\% of packets experience  $>$150ms of \emph{reordering}  delay.}   Indeed, Fig.~\ref{fig:exp:mptcp} shows mild improvements of EDPF  relative to MPTCP in a real experimental setup.
% (details in \S\ref{sec:performance}).

% In this example, two  packets has to be scheduled across two subpaths of identical average one-way delay. A deterministic scheduler that only considers the  expectation of the paths' delays may schedule packet 1 in the ``volatile'' link  forcing packet 2 to await in the receiver's buffer for packet 1 during slow  instantiations of the variable route, thus growing the expected reordering delay. In contrast a stochastic scheduler aware of the random nature of the subpaths would always schedule packet 1  in the low-variable link in this particular example, thus ensuring that none of the packets have to be buffered at the receiver. At larger scales, when many paths and/or packets are to be scheduled,  a deterministic scheduler can lead to high reordering delays, particularly when paths have large variances.

\begin{figure}[t!]
\centering
\includegraphics[width=0.91\columnwidth]{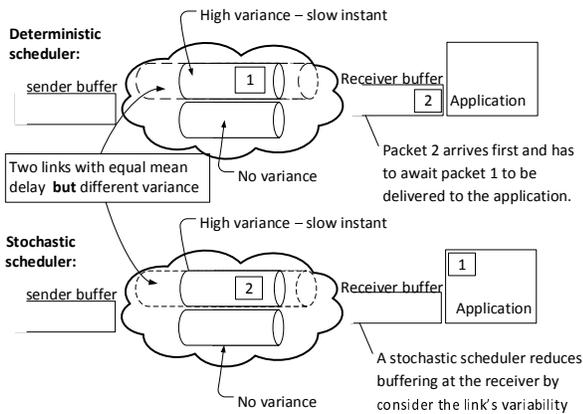}
\vspace{-3mm}
\caption{Two packets are to be scheduled across two subpaths of identical average delay. A ``deterministic'' scheduler may assign packet 1 to subpath 1 (variable delay) and packet 2 to subpath 2 (fixed delay) forcing packet 2 to wait buffered for packet 1 during slow instantiations of subpath 1. In contrast, scheduling packet 1 in subpath 2 would lead to no buffering delays.
% At larger scales, when many paths and/or packets are to be scheduled,  a deterministic scheduler can lead to high reordering delays, particularly when paths have large variance
}
\vspace{-5mm}
\label{fig:example-stochastic}
\end{figure}

% \subsubsection{Losses}
{\bf Losses.} Another source for the HOL problem is losses because  packets have to wait  until these are recovered. ARQ is a well-tested mechanism to address losses, unbeatable when there is no cost for feedback (neither in overhead nor in time). However, high RTTs due to congestion or long distances can severely delay the reception of feedback information and thus the delivery of retransmitted packets. A wealth of approaches based on ARQ or Forward Error Correction (FEC) has been proposed in the past (see~\cite{huang2007improving} and references therein). Recently, the application of network coding in transport protocols has shown a major impact on throughput performance. Network coding was first integrated with TCP by Sundararajan \textit{et al.}\cite{nc-meets-tcp}. SlideOR~\cite{slideor} uses a sliding block mechanism.  CoMP \cite{gheorghiu2010multipath} is a  multipath transport scheme where the rate of coded packets is controlled by a credit-based method. A Fountain (rateless) code is applied to MPTCP in \cite{fountainmptcp} to alleviate bottleneck issues with heterogeneous links. The fundamental problem of all these works, however, is that the throughput gains come at the cost of high delay because  a full set of coded packets need to be received before starting the decoding process (see Fig.~\ref{fig:example-blockcode}). For this reason, we do not consider block or rateless codes in this paper and propose instead a novel streaming coding technique.

\begin{figure}[t!]
\centering
\includegraphics[width=0.91\columnwidth]{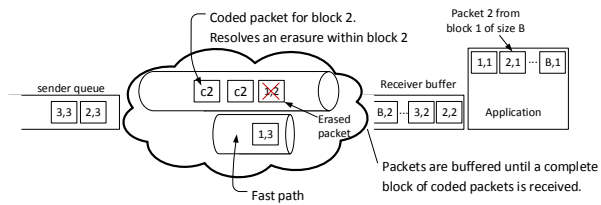}
\vspace{-3mm}
\caption{With block codes, the receiver buffer stores as many \emph{degrees of freedom} (linearly independent coded packets) as the size of the block to decode and deliver data to the application.}
\label{fig:example-blockcode}
\vspace{-5mm}
\end{figure}

\subsection{Our contributions}

Our main contributions are summarised as follows:
\begin{itemize}
 \item We design S-EDPF, a novel scheduler that assigns packets to network interfaces exploiting stochastic information of each subpath's  delays, in contrast to MPTCP's LowRTT or EDPF which do not take such information into account, with the goal of minimising the impact of out of order arrivals in the presence of variable (random) delays. See \S\ref{sec:scheduling}.
 \item S-EDPF features a novel streaming coding scheme that uses lossy links to transmit redundant (coded) information. This mechanism helps us to improve delay performance \emph{dramatically} when there are packet losses caused by wireless noise or interference, and to exploit bad-quality links that otherwise would drag down overall performance. See \S\ref{sec:code}.
 \item We implement S-EDPF in a real prototype. Our prototype is implemented in userspace and works efficiently on multiple platforms (Linux, Android, *BSD, MAC OS X) without the requirement of administration (root) privileges. See \S\ref{sec:prototype}.
 \item We evaluate the performance of our prototype thoroughly in both controlled environments, with emulated conditions to evaluate its behavior,  and \emph{into the wild}, to illustrate the gains of S-EDPF in real environments with real applications. See \S\ref{sec:performance}.
\end{itemize}

Fig.~\ref{fig:exp:mptcp} shows the performance of S-EDPF for 3 different settings (details in the following) demonstrating how delay can be greatly improved in real environments.

\section{Multipath Stochastic Scheduler}
\label{sec:scheduling}

S-EDPF is in charge of selecting which path and at what time to transmit each  packet. In this section we assume lossless channels (the extension to lossy paths is considered in \S\ref{sec:code}) and address the problem of scheduling transmissions with the goal of minimising the impact of out of order arrivals in the presence of random delays.

\subsection{Model Description}

We have to schedule $\mathcal{K}=\{1,2,\dots\}$ packets across $\mathcal{P}=\{1,\dots,P\}$ subpaths with the goal of  minimising in-order delivery delay.  Assume that on each path the time at the sender is slotted, indexed by $\mathcal{S}=\{1,2,\dots\}$, such that one packet can be transmitted in a slot.  Let $t_{p,s}$ denote the start time in seconds of slot $s$ on path $p$.  We do not assume that the slots on different paths are aligned or that slots have fixed duration, and so the link rate of each path may change over time due to the action of congestion control, for instance.
The time at which a packet sent in slot $s$ on path $p$ arrives at the destination is a random variable $a_{p,s}$, with $a_{p,s}\ge t_{p,s}$ to respect causality.  We make the following assumption:
\begin{assumption}[Reordering]\label{assump:reorder}
Consider two slots with indexes $a$ and $b$ on path $p$.   When $a < b$ then $a_{p,a} < a_{p,b}$.   In other words, there  is no reordering of arrivals within the same path.  
\end{assumption}
While packet reordering can occur within a single path (e.g. due to sudden routing changes), it is usually relatively infrequent compared to out-of-order arrivals across paths~\cite{measurements-out-of-order}, 
% (see \cite{fix-out-of-order} for a thorough study on this subject)
and so Assumption \ref{assump:reorder} is mild.   It is important to stress that Assumption \ref{assump:reorder} applies only to packets sent on the \emph{same} path.   Packets sent on different paths may still arrive out of order if they experience different (random) delays.   
Assumption \ref{assump:reorder} implies that the delays experienced by packets on the same path are correlated and not i.i.d. 
% (e.g. a long delay for one packet constraints the next packet to also experience a long delay, else it would arrive out of order).  
To make this explicit, let random variable $\Delta_{p,s,s+1} = a_{p,s+1}-a_{p,s}\ge 0$, $s=1,2,\dots$ and define $\Delta_{p,0,1}=a_{p,1}$.   Then,
\begin{align}
a_{p,s} = \sum_{r=1}^s \Delta_{p,r-1,r}
\end{align}
where $a_{p,1}$ can be computed as $a_{p,1} = \theta_{p,1} + L_{p,1}/B_{p,1}$, where $\theta_{p,1}$,  $L_{p,1}$ and $B_{p,1}$ are the \emph{propagation} delay experienced, the \emph{size} in bits of the scheduled packet, and the access link \emph{bitrate} on slot 1 and path $p$, respectively.

To facilitate scheduling we make the following regularity assumption,
\begin{assumption}\label{assum:iid}
$\Delta_{p,s,s+1}$ is i.i.d., i.e. $\Delta_{p,s,s+1}\!\sim\!{\Delta}_p$.
\end{assumption}
Note that this assumption may be relaxed, at the cost of increasing the complexity of the scheduler.
In addition, we make the following assumptions.

\begin{assumption}\label{assum:noempty}
A packet is transmitted in every slot. 
\end{assumption}
Assumption \ref{assum:noempty} ensures that throughput is maximised.  Note that, by sacrificing throughput, lower delay might be achieved e.g. by  sending packets only along the path with lowest delay.  However, we leave investigation of this type of trade-off between throughput for delay to future work and instead focus on the use of coded packets to trade off throughput for delay (see \S\ref{sec:code}).
\begin{assumption}\label{assum:delay}
Delays are upper bounded by $\bar{T}_p$. 
\end{assumption}
\begin{assumption}\label{assum:slot}
The slot duration on path $p$ can be approximated as being constant, $T_p$, over window $\bar{T}_p$.
\end{assumption}
Assumptions \ref{assum:delay} and \ref{assum:slot} could also be relaxed at the cost of a higher scheduling complexity.

\subsection{Minimum Delay Packet Scheduling}

Let $p_k\in\mathcal{P}$ denote the path on which packet $k\in\mathcal{K}$ is transmitted, and let $s_{k}\in\mathcal{S}$ denote the slot on path $p_k$ in which packet $k$ is transmitted.  Packet $k$ is therefore transmitted at time $t_{p_k,s_k}$ and the time at which packet $k$ arrives is random variable $a_{p_k,s_k}$.  
%Where there is no risk of confusion, we will abbreviate  $X_{i_p,s_p}$ to $X_p$ and $t_{i_p,s_p}$ to $t_p$.

At the receiver we require in-order delivery of packets arriving from multiple paths.  To achieve this, the receiver maintains a reassembly buffer where out of order packets are held until they can be delivered to the application in order.   The delivery time of packet $k$ is therefore the random variable,
\begin{align}
Y_k = \max\{a_{p_1,s_1},a_{p_2,s_2},\dots,a_{p_k,s_k}\}
\end{align}
It will prove useful later to rewrite this expression equivalently as follows. Let $K_{p,k}:=\{q\in\{1,\dots,k-1\}, p_q=p\}$ denote the set of packets sent on path $p$ with indices lower than that of packet $k$ and
\begin{align}
Y_{p,k} : = \max\{a_{p,s_q}: q\in K_{p,k}\}
\end{align}
We then have that
\begin{align}
Y_k = \max\{Y_{1,k},\dots,Y_{P,k},a_{p_k,s_k}\}
\end{align}

Our aim is thus to schedule packet transmissions (i.e. to select path-slot pairs $(p_k,s_k)$, $k=1,2,\dots$) so as to minimise the mean in-order delivery delay
\begin{align}\label{eq:D1}
D:=\lim_{N\rightarrow\infty}\frac{1}{N}\sum_{k=1}^N \mathbb{E}[Y_k] - t_{p_k,s_k}.   
\end{align}

\subsubsection{Low-complexity Scheduling}

% \subsubsection{Sufficiency of In-Order Scheduling on a Path}
However, finding the optimal schedule that minimises eq. (\ref{eq:D1}) is a complex combinatorial problem. Thus, our goal is to design a scheme that solves the problem  with as low computational complexity as possible in order to support the high bitrates expected from a multipath  protocol. We start with the following lemma.

\begin{lemma}\label{lem:one}
Suppose Assumption \ref{assump:reorder} holds.  To minimise delay, it is sufficient to consider situations where packets on the same path are transmitted in order of ascending index.
\end{lemma}
\begin{proof}
	See the appendix.
	\end{proof}

Lemma \ref{lem:one} is intuitive and greatly reduces the set of packet schedules that need to be considered, collapsing this down from all combinations of scheduling orders to which subset of packets are sent on each path.

%%%%%%%%%%%%%%%%%%%%%%%%%%%%%
% \subsubsection{Invariance of Distribution of $Z_{p,k}$}

Now, observe that $Y_{p,k}$ does not depend on the indices of the packets in set $K_{p,k}$, but only on the slots in which the packets are sent.   By Lemma \ref{lem:one} packets are transmitted in ascending order on a path
and by Assumption \ref{assum:noempty} no slots are left unused.  Hence,
\begin{align}
Y_{p,k} = \max\{a_{p,s}, s=1,2,\dots,s_{p,k}\}
\end{align}
where $s_{p,k} = \max \{s_q: q\in K_{p,k}\}$.  
Further, by Assumption \ref{assum:delay} the path $K$ delay has an upper bound $\bar{T}_p$ and so older slots $\{s: t_{p,s}+\bar{T}_p \le t_{p,s_{p,k}}\}$ do not affect the value of $Y_{p,k}$.   It is therefore sufficient to calculate the $\max$ over a finite window of slots, \begin{align}
Y_{p,k} = \max\{a_{p,s}, s=\underline{s}_{p,k},\dots,s_{p,k}\}
\end{align}
with $\underline{s}_{p,k} := s_{p,k}-\delta_{p,k}$ and $\delta_{p,k} = s_{p,k}-\max \{s: t_{p,s}+\bar{T}_p \le t_{p,s_{p,k}}\}$.  
By Assumption \ref{assum:slot} the slot duration can be approximated as being constant, $T_p$.  Hence, for slots $s_{p,k}>\lceil\bar{T}_p/T_p\rceil$ sufficiently far away from the starting slot we have $\delta_{p,k} = \delta_p:=\lceil\bar{T}_p/T_p\rceil$.  That is, $\delta_{p,k}$ is constant and it is \emph{sufficient} to calculate $Y_{p,k}$ over a \emph{fixed} window.  This assumption helps us to manage the computational complexity of the scheduling algorithm but can be readily relaxed e.g using a Chernoff bound it can be shown that the probability that $Y_{p,k}$ depends on events occurring prior to the fixed window is small for an appropriate choice of window size. 
% \tbd{ Done: explain that violating this assumption wouldn't be too bad.}

Recalling random variables $\Delta_{p,\underline{s}_{p,k},\underline{s}_{p,k}+j} = a_{p,\underline{s}_{p,k}+j}-a_{p,\underline{s}_{p,k}}$, $j=0,1,\dots,\delta_p$, we can rewrite $Y_{p,k}$  as
\begin{align*}
Y_{p,k} = a_{p,\underline{s}_{p,k}} \!\!+\!\! \max\{\Delta_{p,\underline{s}_{p,k},\underline{s}_{p,k}+1}, \dots, \Delta_{p,\underline{s}_{p,k},\underline{s}_{p,k}+\delta_p}\}
\end{align*}
Sequence $\mathbf{\Delta}(\underline{s}_{p,k}):=\{\Delta_{p,\underline{s}_{p,k},\underline{s}_{p,k}+1}, \dots, \Delta_{p,\underline{s}_{p,k},\underline{s}_{p,k}+\delta_p}\}$ is, by Assumption \ref{assum:iid}, i.i.d., i.e. $\mathbf{\Delta}(\underline{s}_{p,k})\sim\mathbf{\Delta}_p$.   We have therefore arrived at the following result,
%%%%%%%%%%%%%%%
\begin{theorem}[Invariance of $Y_{p,k} - a_{p,{s}_{p,k}-\delta_p}$]\label{theorem:invariance}
The distribution of $Y_{p,k} - a_{p,{s}_{p,k}-\delta_p}$ is invariant on path $p$ for packets $k$ scheduled in slots $s_{p,k}>\delta_p=\lceil\bar{T}_p/T_p\rceil$.  
\end{theorem}

That is, we can let $Z_p$ describe the distribution of $Y_{p,k} - a_{p,{s}_{p,k}-\delta_p} \sim Z_p\ \forall k$ such that $s_{p,k}>\delta_p=\lceil\bar{T}_p/T_p\rceil$.
Theorem \ref{theorem:invariance} helps us to greatly reduce the complexity of our scheduler as we don't have to recompute ${Z_p}$ for every packet schedule. Moreover, note that the distribution of $a_{p,{s}_{p,k}}$ and $Z_{p}$ can usually be readily estimated in an online fashion.   For example, the realisation of $a_{p,{s}_{p,k}-\delta_p}$ may already be known via ACK feedback from the receiver.  Otherwise, we can use past observations of packet delivery times on path $p$ to estimate the distribution of $a_{p,{s}_{p,k}-\delta_p}$.   
We can also use past observations to estimate the distribution of $\mathbf{\Delta}_p$.   To reduce the number of observations required to estimate the distribution and to reduce computational/memory cost, we can also use a parametric model and estimate the parameters using past observations. For instance, when using a Gaussian model we can find a reasonably accurate approximation of $Z_p$ using \cite{max_approx}.

% When $\Delta_{p,s,s+1}$ is modelled as a Gaussian random variable, we can approximate $Z_p$ as an appropriate Gaussian.  To address this, \cite{max_approx} proposes the ``Partition MBT'' algorithm with a good trade-off between computational complexity and accuracy.\footnote{A thorough study of the approximation errors of this method is covered in \cite{max_approx}.}

\subsubsection{Stochastic Earliest Delivery Path First}
Using Theorem \ref{theorem:invariance}, we can rewrite eq.~(\ref{eq:D1}) as
\begin{align*}
% D&=\lim_{N\rightarrow\infty}\frac{1}{N}\sum_{k=1}^N \mathbb{E}[Y_k] - t_{p_k,s_k}\\
% &= \lim_{N\rightarrow\infty}\frac{1}{N}\sum_{k=1}^N  \mathbb{E}[\max\{Y_{1,k},\dots,Y_{P,k},a_{p_k,s_k}\}]- t_{p_k,s_k}\\
% &= \lim_{N\rightarrow\infty}\frac{1}{N}\sum_{k=1}^N  \mathbb{E}[\max\{Z_{1}+a_{1,s_{1,k}-\delta_1},\dots,\\
% &~~~~~~~~~~~~~~~~~~~~~~~~~Z_{P} + a_{P,s_{P,k}-\delta_P},a_{p_k,s_k}\}]- t_{p_k,s_k}\\
D&= \lim_{N\rightarrow\infty}\frac{1}{N}\sum_{k=1}^N  \mathbb{E}[\max\{Z_{1}+a_{1,s_{1,k}-\delta_1}- t_{p_k,s_k},\dots,\\
&~~~~~~~~~~~~~~~~~~~~~~~~~Z_{P} + a_{P,s_{P,k}-\delta_P}- t_{p_k,s_k},D_{p_k,s_k}\}]
%&= \lim_{N\rightarrow\infty}\frac{1}{N}\sum_{p=1}^N  \mathbb{E}[\max\{Z_{1}+D_{1,s_{1,p}-\delta_1}+t_{1,s_{1,p}-\delta_1}- t_{p_k,s_k},\dots,Z_{m}+D_{m,s_{m,p}-\delta_m}+t_{m,s_{m,p}-\delta_m}- t_{p_k,s_k},D_{i_p,s_p}\}]
\end{align*}
where $D_{p,s}:=a_{p,s}-t_{p,s}$. 
By Assumption \ref{assum:noempty} and Lemma \ref{lem:one}, minimising $D$ is equal to minimising 
\begin{align*}
&\lim_{N\rightarrow\infty}\frac{1}{N}\sum_{k=1}^N  \mathbb{E}[\max\{Z_{1}+a_{1,s_{1,k}-\delta_1},\dots,\\
&~~~~~~~~~~~~~~~~~~~~~~~~~Z_{P} + a_{P,s_{P,k}-\delta_P},a_{p_k,s_k}\}] 
\end{align*}
and the minimum delay schedule is to greedily select $(p_k,s_k)$ to minimise $\mathbb{E}[\max\{Z_{1}+a_{1,s_{1,k}-\delta_1},\cdots,Z_{P} + a_{P,s_{P,k}-\delta_P},a_{p_k,s_k}\}]$.   We thus propose  Stochastic Earliest Delivery Path First (S-EDPF) to  schedule the transmission of packet $k$ in the first available slot of such path $p_k^*$  with minimum expected reordering latency, i.e.,
%  \tbd{need to prove this?}
% The main difficulty is the dependence on $a_{p,s_{p,k}-\delta_p}$, $p\in\{1,\dots,P\}$ which acts to couple scheduling decisions.  
%  One possible approach is to assume that
% \begin{align}
% X_{i,s_{i,p}-\delta_i}- t_{i_p,s_p}=D_{i,s_{i,p}-\delta_i}+t_{i,s_{i,p}-\delta_i}- t_{i_p,s_p} = t_i
% \end{align}
% where $t_i$ is constant.  For example, this holds when random delay $|D_{i,s_{i,p}-\delta_i} - E[D_{i,s_{i,p}-\delta_i}]| \ll t_{i,s_{i,p}-\delta_i}- t_{i_p,s_p}$ \textbf{**likely to be true ?} and $\delta_i$ is sufficiently large that $t_{i,s_{i,p}-\delta_i}- t_{i_p,s_p}$ can be treated as constant (slot quantisation effects can be neglected).  Then,
% \begin{align*}
% D&= \lim_{N\rightarrow\infty}\frac{1}{N}\sum_{p=1}^N  E[\max\{Z_{1}+t_1,\cdots,Z_{m}+t_m,D_{i_p,s_p}\}]
% \end{align*}
% and the minimum delay schedule is to greedily select $(i_p,s_p)$ to minimise $E[\max\{Z_{1}+t_1,\cdots,Z_{m}+t_m,D_{i_p,s_p}\}]$ subject to the constraint that $(i_p,s_p)$ satisfies Assumption \ref{assum:noempty} (no empty slots) and Lemma \ref{lem:one} (in-order transmission on each path).  \textbf{**need to prove this, not just state it}
% Alternatively, if there exist renewal times where the $X_{i,s_{i,p}-\delta_i}$, $i\in\{1,\cdots,m\}$ couple, then we could determine schedules over frames corresponding to the renewals.   
\begin{align}\label{eq:scheduler}
  &p_k^*  = \argmin_{p_k}  \mathbb{E}[\max\{Z_{1}+a_{1,s_{1,k}-\delta_1},\dots,\nonumber\\
  &~~~~~~~~~~~~~~~~~~~~~~~~~Z_{P} + a_{P,s_{P,k}-\delta_P},a_{p_k,s_k}\}] 
\end{align}
\section{Low Delay Streaming Code}
\label{sec:code}

% With the goal of minimising delay when losses occur, increase resiliency, and further improve performance in the presence of highly variable links, we design a novel Forward Error Correction (FEC) scheme for S-EDPF.
In the event of a loss, which is frequent in wireless links, the usual recovery method, ARQ,  costs extra delay because of the round trip time that it takes for the repeat request to be delivered to the transmitter and the message to be retransmitted to the receiver. %This type of delay is different in nature from the scheduling delay caused by  reordering at the receiver. 
% The scheduler designed in the previous section proposes an algorithm that minimises in order delivery delay. 
Neither LowRTT (MPTCP) nor EDPF use loss information to schedule non-retransmitted packets, though they both preferentially send data on links with higher rate. If these links were lossy, these scheduling choices could do more harm than good to the overall performance. Our approach is to precode enough information in coded packets to anticipate losses and help the receiver recover such packets without a need for retransmission.

% As we explain in the sequel, S-EDPF is specifically designed to exploit part of the bandwidth of lossy links to jointly schedule information and coded packets with the goal of minimising in-order delivery delay.

% In this section we propose a novel low-delay random linear coding technique to precode  enough information in coded packets to combat losses.  in contrast to other  approaches, perform a joint scheduling of information and coded packets.

% In this way, in contrast to other approaches, S-EDPF uses part of the bandwidth of lossy links to precode and schedule jointly with information packets

% design a low delay random linear coding technique that anticipates losses and precodes enough information in coded packets which are jointly scheduledscheduled over multiple links to help the receiver to recover from losses without a need for retransmission. 

% \tbd{Done: Sentence to illustrate why a FEC scheme like ours would help to reduce delay and increase resiliency when there are losses and highly variable links}

% With the purpose of designing an efficient FEC mechanism, we use a different model of the system. 
Similarly as above, we divide time into slots $\mathcal{S}=\{1,2,\dots\}$ each corresponding to the transmission of one packet. S-EDPF generates a \emph{coded packet} $c_{k}$ and schedules it every $\tau$ slots (which we refer to as the \emph{coding interval}). A coded packet $c_k$  is a random linear combination of information packets 1 to $N_\tau k$, where $N_\tau = \tau-1$ is the number of uncoded packets $u_i$ sent between $c_k$ and $c_{k-1}$ across all subpaths, i.e., 
\begin{align}
c_k = f_{ \tau}(u_1, u_2 , \dots ,u_{N_\tau \cdot k } ) := \sum_{j=1}^{N_\tau\cdot k} w_{k j} \cdot u_j
\end{align}
with coefficients $w_{k j}$, selected identically and independently, uniformly at random from a finite field of size $Q$.  At the receiver, upon reception of $c_k$, a matrix $G_k$ is constructed with rows formed from the coefficients of the received packets (normal uncoded packets adding a row with a $1$ in the corresponding diagonal, all other entries being $0$).  
% This simple, yet effective, coding scheme relies on packets arriving \emph{in order} to the decoder. Thus, its application with the packet scheduler proposed in \S\ref{sec:scheduling}, that minimises the chances of out-of-order arrivals, is crucial. 
Decoding can then be carried out on-the-fly using e.g. Gaussian elimination.   In our analysis we will make the standing assumption that the field size $Q$ is sufficiently large that with probability one each coded packet helps the receiver recover from one information packet erasure.  That is, each coded packet row added to generator matrix $G_k$ increases the rank of $G_k$ by one. Note however that we evaluate a real decoder (with a non-null decoding failure probability) in \S\ref{sec:performance}. 
% In our prototype, we select $Q=256$, which proves to be a good trade-off between complexity and low decoding failure probability.

In this way, the task of S-EDPF is to \emph{jointly schedule} the transmission of coded and information packets leveraging the path selector proposed in \S\ref{sec:scheduling} (that minimises out-of-order arrivals).

\subsection{Buffering Delay Analysis}
\label{sec:delay-analysis}

\begin{figure}[t!]
\includegraphics[width=\columnwidth]{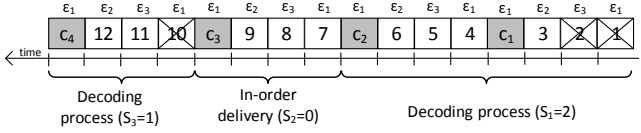}
\vspace{-8mm}
\caption{Coding scheme.}
\label{fig:coding-scheme}
\vspace{-5mm}
\end{figure}

In this section, we characterise the buffering delay performance of this coding scheme when delays are fixed.\footnote{We assess this coding scheme in the presence of random delays experimentally in \S\ref{sec:performance} and leave the mathematical analysis for future work.}
% We thus assume that packets arrive in order to the decoder (or that have been re-ordered before arriving to the decoder). 
% To make the analysis tractable, in this section we measure delay in slot units.\footnote{Note that this simplification skips the fact that different slots may have different duration (e.g. because they belong to heterogeneous subpaths), yet it gives enough granularity for S-EDPF to decide on the coding interval.}
% \footnote{Although this model neglects the fact that different packets have different propagation/transmission delays and that multiple packets are transmitted simultaneously, such delays are already considered in the analysis of \S\ref{sec:scheduling}. In this section, we focus only on the additional delay imposed by the loss-recovery mechanism.}. 
Information packets are delivered to the application in order as they arrive (without buffering delay), until there is a loss. We name this state as ``in-order delivery''. When there is a loss, a ``decoding process'' starts: packets are buffered until the decoder has enough coded packets to fill the gaps (losses), at which point \emph{in-order delivery} resumes. Following \cite{karzand-2014}, we index each ``in-order delivery''/``decoding process'' period with $j=1,2,\dots$, and let the random variable $S_j$ count the coded packets required in $j$ to resume an in-order delivery state ($S_j=0$ if stage $j$ is already in the in-order delivery state). The above is illustrated in Fig.~\ref{fig:coding-scheme} for $\tau=4$ and packets being received from $P=3$ subpaths with erasure probability $\epsilon_1$, $\epsilon_2$, and $\epsilon_3$, respectively.

\begin{theorem}[S process]\label{theorem:S-process}
Suppose we transmit information packets across $\mathcal{P}=\{1,\dots,P\}$ independent subpaths with erasure probability $\{\epsilon_1,\dots,\epsilon_P\}$ and fixed  delays. Suppose we insert a coded packet on path $p_c\in\mathcal{P}$ (with erasure probability $\epsilon_{p_c}$) in between every $N_\tau=\tau-1$ information packets. Assume that each coded packet can help us to recover from one erasure, irrespective of the subpath where the erasure has occurred.  Denote by $N_{\tau,p}$ the number of packets sent in subpath $p$ between two coded packets (including last coded packet if sent on $p$), i.e.,  $\sum_{p\in\mathcal{P}} N_{\tau,p}=\tau$. 
Then, we have: 
% \tbd{assumption on slot alignment across paths}

\begin{enumerate}

\item For all $\epsilon_p$ and $N_{\tau,p}$ such that $\sum_{p\in\mathcal{P}} N_{\tau,p} \epsilon_p < 1$, the mean of the probability distribution of $S$ exists and is finite. 

\item The distribution of $S$ is characterised by:
\begin{align}
P(S\!\!=\!\!0) & = (1\!\!-\!\!\epsilon_{p_c})^{N_{\tau,p_c}\!-\!1} \prod_{i\ne p_c} (1\!\!-\!\!\epsilon_i)^{N_{\tau,i}} &\\ 
P(S\!\!=\!\!1) & = (N_{\tau,p_c}\!\!-\!\!1)\epsilon_{p_c}(1\!\!-\!\!\epsilon_{p_c})^{N_{\tau,p_c}\!-\!1}\prod_{j\ne p_c}(1\!\!-\!\!\epsilon_j)^{N_{\tau,j}} + \nonumber\\
       & + \sum_{i \ne p_c}  N_{\tau,i}\epsilon_i(1-\epsilon_i)^{N_{\tau,i}\!-\!1} \prod_{j \ne i} (1-\epsilon_j)^{N_{\tau,j}}  \\
P(S\!\!=\!\!k) & \!\approx\!  \frac{N_\tau}{k}    \bar\epsilon^{k} (1\!\!-\!\!\bar\epsilon)^{kN_\tau} {(k-1)\tau \choose k-1},\ \forall k\!>\!1 \nonumber \label{eq:distrib_S:tail}\\
\end{align}
where $\bar\epsilon = \frac{\sum_{p\in\mathcal{P}} N_{\tau,p}\epsilon_p}{ N_\tau }$.

% \item
% \begin{align}
% & \frac{\sum_{i=1}^{N_\tau} (\bar\epsilon -  \epsilon_i)^2}{N_\tau \bar\epsilon(1-\bar\epsilon)}\min(1,N_\tau\bar\epsilon(1-\bar\epsilon)\frac{1}{124} \le \nonumber \\
% & \sup \left | P(S=k) - P(S_1=k) \right |  \le \\ 
%  & \frac{1-\bar\epsilon^{N_\tau+1} - (1-\bar\epsilon)^{N_\tau+1}}{(N_\tau+1)\bar\epsilon(1-\bar\epsilon)}\sum_{i=1}^{N_\tau} (\bar\epsilon - \epsilon_i)^2,\ \forall k>1  \nonumber
% \end{align}

\item The first and second moments of $S$ can be approximated by the following closed-form expressions:
\begin{align}
\!\!\!\!\mathbb{E}[S]\!\! &\approx\!\! P(S\!\!=\!\!1)\!\! + \!\!\frac{\tau(N_\tau)\bar\epsilon^2(1\!\!-\!\!\bar\epsilon)^{N_\tau}}{1\!\!-\!\!\tau\bar\epsilon} &\label{eq:distrib_S:1st_moment}\\
\!\!\!\!\mathbb{E}[S^2]\!\! &\approx\!\! P(S\!\!=\!\!1)\!\! +\!\!(1\!\!-\!\!\bar\epsilon\!\!+\!\!(1\!\!-\!\tau\bar\epsilon)^2)\frac{\tau N_\tau\bar\epsilon^2(1\!\!-\!\!\bar\epsilon)^{N_\tau}
}{(1\!\!-\!\!\tau\bar\epsilon)^3} \label{eq:distrib_S:2nd_moment}
%        &+\frac{N_\tau(N_\tau \!\!-\!\!1)\bar\epsilon^2(1\!\!-\!\!\bar\epsilon)^{N_\tau\!-\!1} ((1-N_\tau\bar\epsilon)^2-\bar\epsilon)}{(1-\N_\tau\bar\epsilon)^3}
% \mathbb{E}[S] &\approx \frac{(N_\tau-1) \bar\epsilon (1-\bar\epsilon)^{N_\tau-1}}{1-N_\tau \bar\epsilon} &\\
% E(S^2) &\approx \mathbb{E}[S] + \frac{N_\tau(N_\tau-1) \bar\epsilon^2 (1- \bar\epsilon)^{N_\tau} }{(1-N_\tau \bar\epsilon)^3}
\end{align}
where $\bar\epsilon= \frac{\sum_{p\in\mathcal{P}} N_{\tau,p}\epsilon_p}{ N_\tau }$.

\end{enumerate}
\end{theorem}
\begin{proof}
	See the appendix.
\end{proof}

%  \subsubsection{Approximation used in Theorem \ref{theorem:S-process}}\label{sec:code:approx}
% \tbd{this subsection is a good candidate to be removed is space is needed -- this approximation is indirectly evaluated in the next subsection}
In this theorem we approximate \emph{the tail} of the distribution of $S$ (see eq.~(\ref{eq:distrib_S:tail})) because the exact solution requires running a complex iterative method to compute probabilities from a poisson-binomial distribution~\cite{ehm-poisson_binomial}.  In more detail, we approximate a poisson-binomial distribution with a binomial distribution (see the proof of Theorem~\ref{theorem:S-process}.2) which is more tractable, although not accurate in general. However, we find that it provides a good approximation of the complete distribution of $S$ because we only apply it to its tail (i.e. in $P(S=k),\ \forall k>1$), which does not contribute much to the whole distribution, i.e., $\sum_{k=2}^\infty P(S=k) \ll \sum_{k=0}^{1}P(S=k)$, in most of the cases of interest. We confirm its accuracy by means of simulations below. 
% 
% To evaluate this approximation, we set up a scenario with 2 subpaths of equal rate and packet erasure probabilities $\epsilon_1$ and $\epsilon_2 = \{0, 0.2\epsilon_1, 0.4\epsilon_1, 0.6\epsilon_1, 0.8\epsilon_1, \epsilon_1 \}$, respectively, and compute the first and second moments for coding intervals $\tau=\{2, 4, 8\}$ with $(i)$ the approximated expression proposed in eq.~(\ref{eq:distrib_S:1st_moment}) and (\ref{eq:distrib_S:2nd_moment}), and $(ii)$ the exact solution using \cite{poisson-binomial-closed} to compute the poisson binomial $\mathcal{L}$ of eq.~(\ref{eq:exact_S}).  The coded packets are transmitted on the path with lowest erasure probability.\footnote{As we will see later on, this is actually the best choice.}  Fig.~\ref{fig:approx:E_S} show $\mathbb{E}[S]$ and $\mathbb{E}[S^2]$ as a function of $\epsilon_1$ and shows, as expected, that the approximation is particularly accurate for large coding intervals and small $|\epsilon_1-\epsilon_2|$'s, but it is also good in the extreme cases, e.g. with highly heterogeneous links. 
% 
% 
% 
% \begin{figure}
% \centering
% \includegraphics[width=\columnwidth ]{approximation.eps}
% \vspace{-15mm}
% \caption{$\mathbb{E}[S]$ and $\mathbb{E}[S^2]$ for different coding intervals $\tau$ and packet loss probabilities. Exact values vs. approximations. }
% \label{fig:approx:E_S}
% \vspace{-5mm}
% \end{figure}

 \subsubsection{Buffering delay}

Under the assumption of fixed delays, we can now characterise the in-order delivery delay of an optimal earliest arrival delivery scheduler using Theorem \ref{theorem:S-process}. 
Let us define a \emph{frame} as the transmission of $N_{\tau}$ information packets plus a coded packet. The scheduler in each frame assigns $N_{\tau, p}$ packets to each subpath $p$ such that  $\tau = \sum_{p\in\mathcal{P}} N_{\tau,p}$.  
The relationship between the S process and  buffering delay due to losses is as follows.

\begin{theorem}[Buffering delay]\label{theorem:delayS}
	At the receiver, the asymptotic mean buffering delay per transmitted packet is upper bounded by $$\frac{\mathbb{E}[S^2]}{ 2{\tau}({\mathbb{E}}[S] + P(S=0))} \sum_{p\in\mathcal{P}}  \max \left \{ N_{\tau,p}(N_{\tau,p}-1), 1 \right \} \Delta_p$$
	time units.

\end{theorem}

\begin{proof}
See the appendix.
\end{proof}

To illustrate the above, consider a scenario with 2 subpaths having equal rate and packet erasure probabilities $\epsilon_1$ and $\epsilon_2 = \{0, 0.2\epsilon_1, 0.4\epsilon_1, 0.6\epsilon_1, 0.8\epsilon_1, \epsilon_1 \}$, respectively. First, we compute the bound on packet buffering delay given by Theorem~\ref{theorem:delayS} for different coding intervals; second, we simulate these scenarios with a custom event-driven simulator and measure the mean buffering delay; and finally we compare both results in Fig.~\ref{fig:approx:delay}. From the figure, we conclude that  Theorem~\ref{theorem:delayS} predicts buffering delay tightly and therefore it also helps to validate the approximations used in Theorem \ref{theorem:S-process}.

\begin{figure}[t!]
\centering
\vspace{-5mm}
\includegraphics[width=\columnwidth ]{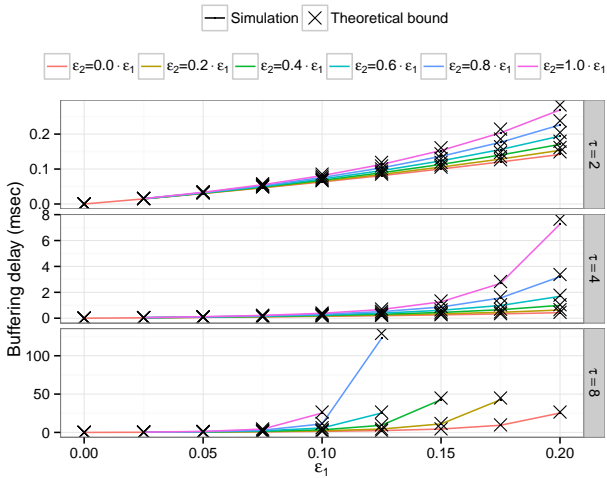}
\vspace{-10mm}
\caption{ Buffering delay. Simulation results \emph{vs.} theoretical upper bound (Theorem~\ref{theorem:delayS}). }
\vspace{-5mm}
\label{fig:approx:delay}
\end{figure}

%  \begin{figure}[t]
%         \centering
%              \subfloat[Campus.]{
% 	    \includegraphics[width=0.33\columnwidth]{delay_bound_2.eps}
% 	    \label{fig:approx:delay:2}
%         }%
%         \subfloat[Home.]{
%                 \centering
% 	    \includegraphics[width=0.33\columnwidth]{delay_bound_4.eps}
% 	    \label{fig:approx:delay:4}
%         }
%         \subfloat[]{
%                 \centering
% 	    \includegraphics[width=0.33\columnwidth]{delay_bound_8.eps}
% 	    \label{fig:approx:delay:8}
%         }       
% \caption{ Buffering delay. Simulation results vs. theoretical upper bound (eq.~(\ref{theorem:delayS})). }
%   \label{fig:approx:delay}
% \end{figure}

\subsection{Path selection for coded packets}

The above is a very useful tool for taking decisions as to when to schedule transmission of coded packets. In addition, the next result tells us how we can best exploit the degree of freedom given by multiple subpaths.

\begin{theorem}[Path choice for coded packets]\label{theorem:coded-path}
Under the conditions of Theorem~\ref{theorem:S-process}, for a given $N_{\tau,p}$,  buffering delay is minimised when coded packets are scheduled in such path $p_c$ with highest erasure probability.
% \tbd{what if these assumptions are violated? Mohammad: Which one? we assumed the path is modeled as an erasure channel. That is the only assumption I think we made.}

\end{theorem}
\begin{proof}
	See the appendix.
\end{proof}

The importance of Theorem~\ref{theorem:coded-path} is considerable because it allows us to exploit bad-quality links to further improve overall performance and support a high degree of resiliency in our system. This contrasts to other approaches, such as MPTCP or EDPF, where multipath transportation need not provide any gain in performance (even a loss in some cases) over its single-path counterpart when some of the subpaths are lossy~\cite{chen2013measurement}. 

It can be easily shown that the delay gain from applying Theorem \ref{theorem:coded-path} (as opposed to scheduling coded packets on the better link) is a multiple of $P(S=1) \tau$. The value of $\tau$ is bounded by the capacity of the channel but, as $\tau$ becomes larger, $P(S=1)$ also decreases because the decoding events are becoming longer. Our simulations (which we do not show here due to space constraints) also confirm this intuition and show that for small values of $\tau$ (far below the channel capacity) the gain in delay is of the same order of magnitude of the mean buffering delay; however, as the operating rate of the system gets closer to the channel capacity, the effect of the placement of the coded packet vanishes.

% \tbd{Validation through simulations in Fig.~\ref{fig:approx:delay_coding_path}.}
% \begin{figure}[t!]
% \centering
% \includegraphics[width=0.85\columnwidth ]{delay_coded_link.eps}
% \vspace{-5mm}
% \caption{ Buffering delay. Coding across lossiest link \emph{vs.} less lossy link. Simulation results. }
% \label{fig:approx:delay_coding_path}
% \end{figure}

\section{Prototype Implementation}\label{sec:prototype}

We use the framework of Coded TCP (CTCP)\cite{ctcp} to implement S-EDPF. CTCP is composed of a pair of SOCKS5 proxies that route data from/to TCP applications into/from (multiple) UDP standard sockets. This gives us the ability to easily implement and evaluate congestion/rate control, coding/decoding and scheduling algorithms \emph{in userspace}. This not only facilitates research and development but it also maximises  its deployability: S-EDPF/CTCP runs in Linux, OS X, Android and *BSD without administration privileges.   
% Fig.~\ref{fig:architecture:sender} and~\ref{fig:architecture:receiver} illustrate the architecture of the sender and receiver sides, respectively.

%  \begin{figure}[tb!]
%         \centering
% 	    \includegraphics[width=1\columnwidth ]{ctcp_sender.eps}
%        \caption{Prototype architecture. Sender.}
%   \label{fig:architecture:sender}
%  \end{figure}
%  
% \begin{figure}[tb!]
% \centering
%   \includegraphics[width=1\columnwidth]{ctcp_receiver.eps}
%          \caption{Prototype architecture. Receiver.}
% 
%   \label{fig:architecture:receiver}
%   
% \end{figure}

The encoder/decoder is implemented using \emph{finite field arithmetic} in $GF(256)$. This lets us encode and decode information efficiently,  using XORs for addition and subtraction and quick table lookups using SIMD programming for multiplication, without compromising performance (i.e. with very low decoding error probability). 
%At the receiver, each received packet adds a row to a generator matrix with the coefficients of the coded packets. Uncoded packets add a row with a single 1 in the corresponding diagonal, all other entries being 0. In this way, both the generator matrix and the payload of each received packet form a system of linear equations. 
%To decode (solve the linear system),   
Decoding is based on a Gaussian elimination algorithm that is called for every received packet and a back substitution algorithm that is called when there are sufficient degrees of freedom (enough linearly independent equations). Decoded packets are released to the TCP socket as soon as they can be handed in order.

To assign packets into each subpath's queue as described in \S\ref{sec:scheduling}, we leverage rate and delay information obtained from feedback to compute a Gaussian approximation of $Z_p$  every half a second using the ``Partition MBT'' algorithm proposed in \cite{max_approx}.
% \footnote{Approximation errors of this technique are studied in\cite{max_approx}.} 
Then, for every incoming packet, we use the last computation of $\vv{Z}$ to solve eq. (\ref{eq:scheduler}) and enqueue each packet accordingly.

We also implement a simple Selective ARQ mechanism to assist our coding scheme. Received packets trigger the transmission of cumulative \emph{acknowledgements} (ACKs) that give the transmitter information regarding delivery delays (used by the packet scheduler), packet loss rate (used by the coding scheme), which packets have been lost (used by the ARQ mechanism), and congestion information (used for rate control).

Finally, unless otherwise stated, we use CTCP's default congestion control algorithm, which is particularly suitable for the sort of networks we target (lossy/variable) and which is fair to other TCP flows (see \cite{ctcp}).

\section{Performance Evaluation}\label{sec:performance}

In the sequel, we summarise a thorough experimental evaluation in both controlled and real environments.  First, we validate our prototype by comparing its performance with the Linux implementation of MPTCP; secondly, we run a set of experiments in controlled environments emulating network conditions; and finally, we run an experimental campaign ``in the wild''.
% Our testbed is displayed in Fig.~\ref{fig:exp:testbed}. 

% First, we run experiments in a \emph{controlled} environment, using an internal high-speed Ethernet network, to evaluate the behavior of our prototype under different (controlled) network conditions. Secondly, we set up a 3G/4G access interface with the Three Ireland operator and a IEEE 802.11 interface which we connect to different WiFi hotspots. These experiments are aimed at evaluating the gains of our prototype ``in the wild''.

 \begin{figure}[t!]
\centering
\includegraphics[width=\columnwidth ]{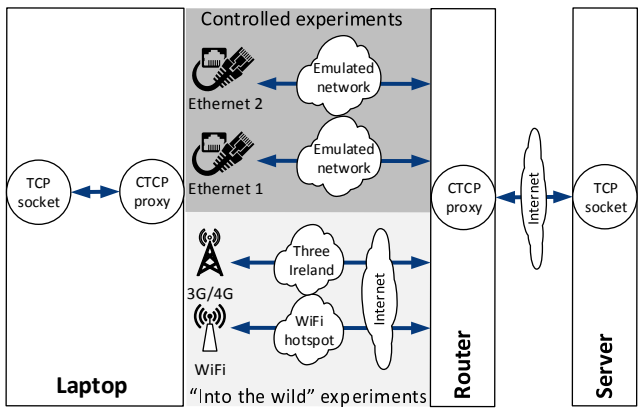}
\vspace{-7mm}
\caption{Testbed configurations.}
\vspace{-5mm}
\label{fig:exp:testbed}
\end{figure}

\subsection{Comparison With MPTCP}

Fig.~\ref{fig:exp:mptcp} anticipated a comparison between MPTCP, EDPF and S-EDPF-$\tau$ with different coding intervals $\tau$.   This experiment was carried out in a home environment, uploading data for 30 s from a laptop attached to a 2.4-Ghz WiFi Access Point and a Meteor 3G/4G dongle to a remote server using {\ttfamily iperf}. For ``Legacy MPTCP'', we used the Linux implementation of MPTCP v0.89 with OLIA~\cite{khalili2013mptcp} for congestion control and  LowRTT as scheduler. We repeat the experiment 5 times for every scheme. The first conclusion is that MPTCP shows the worst performance in terms of both throughput and delay. This could be explained due to MPTCP using a different congestion control scheme (less friendly to wireless losses as shown in \cite{ctcp}), a packet scheduler (LowRTT) that does not account for delay variability (nor does EDPF), and the lack of FEC for packet loss control. The second conclusion is that  S-EDPF ($i$) increases the capacity of the multipath network (due to a more efficient scheduling of transmissions) and ($ii$) is capable of trading data throughput for large delay improvements by tuning $\tau$ (due to our coding scheme).
% very similarly to EDPF, which is explained by the fact that neither of them exploit stochastic knowledge of the subpaths nor use FEC. In contrast, S-EDPF-$\tau$ shows substantial gains in delay (reducing it by half in the least aggressive configuration) and, interestingly, in throughput as well (due to the more efficient use of the links).  
% In light of this result, in order to provide a fair comparison between S-EDPF and other approaches (i.e. using the same congestion control algorithm), we implement MPTCP's LowRTT scheduling scheme within our prototype and use it as a benchmark later on.

 \subsection{Experiments on a Controlled Environment}

 We use {\ttfamily tc/netem} to generate normally distributed delay samples which are not correlated\footnote{This is particularly hostile for S-EDPF whose design assumes that delays in the same path are correlated.} and {\ttfamily dummynet} to emulate link bitrates and drop packets randomly when required.  Finally, to avoid external effects from congestion control (e.g. slow start), in this subsection we fix the contention window of each subflow to the bandwidth-delay product of each link (i.e. highest rate without causing congestion). 
 Unless otherwise stated, we send data using {\ttfamily iperf} for 30 s from a laptop attached to two emulated networks, as depicted in Fig.~\ref{fig:exp:testbed}, and repeat each experiments 5 times per mechanism.  
 In the sequel, we first evaluate one variable at a time (randomness, then losses) and then we test the performance of a real video streaming service. 
%  Note that we assess the impact of the randomness and loss experienced in production networks in \S\ref{sec:performance:wild}.

\subsubsection{Random propagation delays}

\begin{figure}[t!]
\centering
\vspace{-8mm}
\includegraphics[width=\columnwidth ]{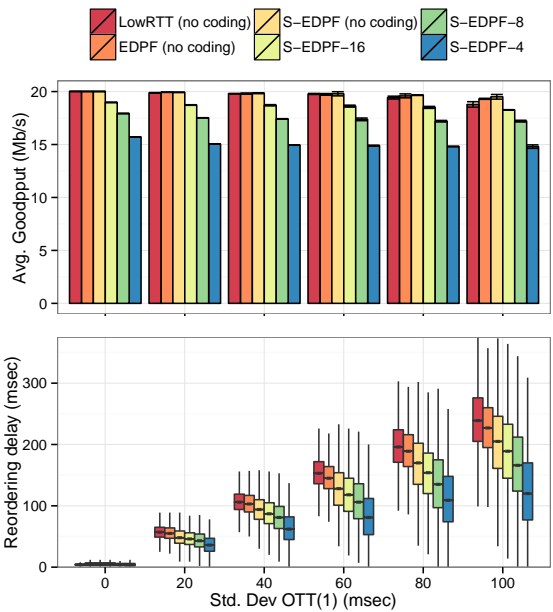}
\vspace{-17mm}
\caption{Buffering delay and throughput upon random (controlled) delays and no loss.}
\vspace{-5mm}
\label{fig:exp:control_1}
\end{figure}

We set up two links at 10~Mb/s each connected to a router. The router introduces  50~ms of fixed propagation delay in link 2 and a random propagation delay with mean 50~ms  and variable standard deviation in link 1. We don't drop any packet in these experiments. We compare MPTCP's LowRTT scheduling scheme, EDPF and S-EDPF-$\tau$ (i.e. S-EDPF with coding interval $\tau$). Results are shown in Fig.~\ref{fig:exp:control_1}. The top subplot shows the average goodput performance experienced at the application layer and the bottom subplot whiskers and boxes for the measured buffering delay (i.e. not including the link's propagation delay). Note that we fully use the aggregate capacity of the network, i.e., the receiver always receives ~20Mb/s of raw traffic. In this way, this figure illustrates how our FEC mechanism uses part of this capacity to send redundant information; for instance, S-EDPF-4 trades 5~Mb/s to reduce delay by half (roughly) when the standard deviation of link 1's propagation delay is larger. 
Note that coding helps delay even when there is no loss. The reason is that depending on the behaviour of the channels, coded packets might arrive sooner than preceding information packets and thus help the receiver to decode still-on-the-air packets sooner than the case when there is no coding. We have evaluated this scenario with different mean delays and S-EDPF offers similar gains in performance (not shown here due to space constraints).
% \tbd{why does coding help delays when no loss?M: Done}

\subsubsection{Lossy subpaths}

We now evaluate the performance of S-EPDF in the presence of lossy paths. To this aim, we fix the propagation delay of two links to 50~ms and the bandwidth at 10~Mb/s, and randomly drop $10\%$ of packets in link 1. In link 2, we don't drop any packets in the experiments depicted in Fig.~\ref{fig:exp:control_2:0}, and $10\%$ of packets in Fig.~\ref{fig:exp:control_2:10}. Note that we only drop data packets and not acknowledgements, to benchmark against an ideal ARQ mechanism. In the figures, we compare our S-EDPF scheduler with different coding intervals and a scheme that only relies on retransmissions (ARQ) to handle losses, like MPTCP. We use EDPF to schedule packets in ``ARQ''.  As depicted in the figure, the scheduling of FEC coded packets helps us to practically eliminate buffering delay with the most aggressive configuration (S-EDPF-4) at the cost of ~15\% of throughput. It is worth mentioning that higher RTTs show larger gains in performance (not shown here for space reasons) because ARQ only retransmits packets upon the reception of feedback and this takes longer the higher the propagation delay.

 \begin{figure}[t!]
 \vspace{-5mm}
        \centering
             \subfloat[Heterogeneous links.]{
	    \includegraphics[trim = 0mm 60mm 0mm 0mm, clip, width=0.5\columnwidth ]{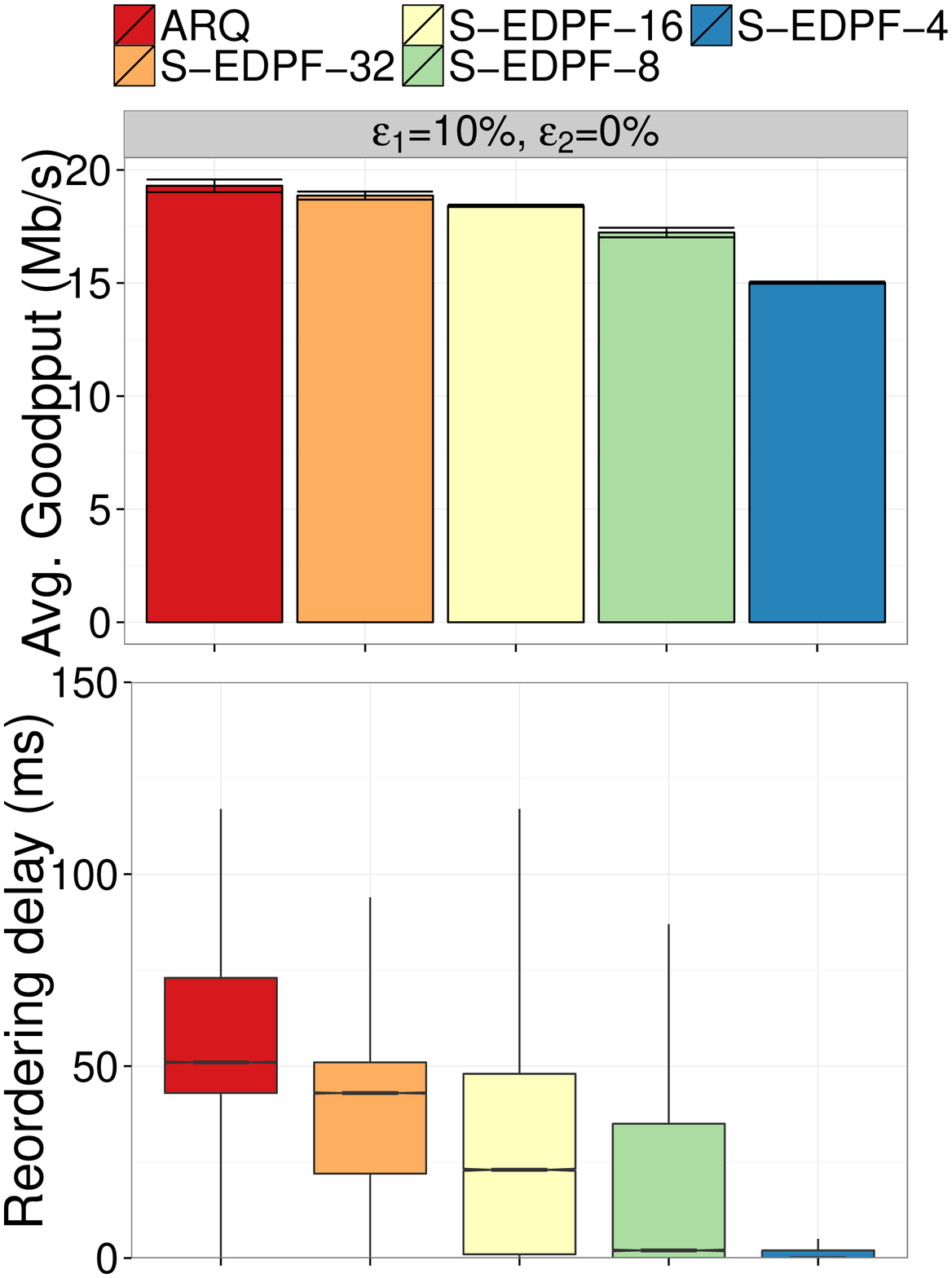}
	    \label{fig:exp:control_2:0}
        }%
        \subfloat[Homogeneous links.]{
                \centering
	    \includegraphics[trim = 0mm 60mm 0mm 0mm, clip,width=0.5\columnwidth]{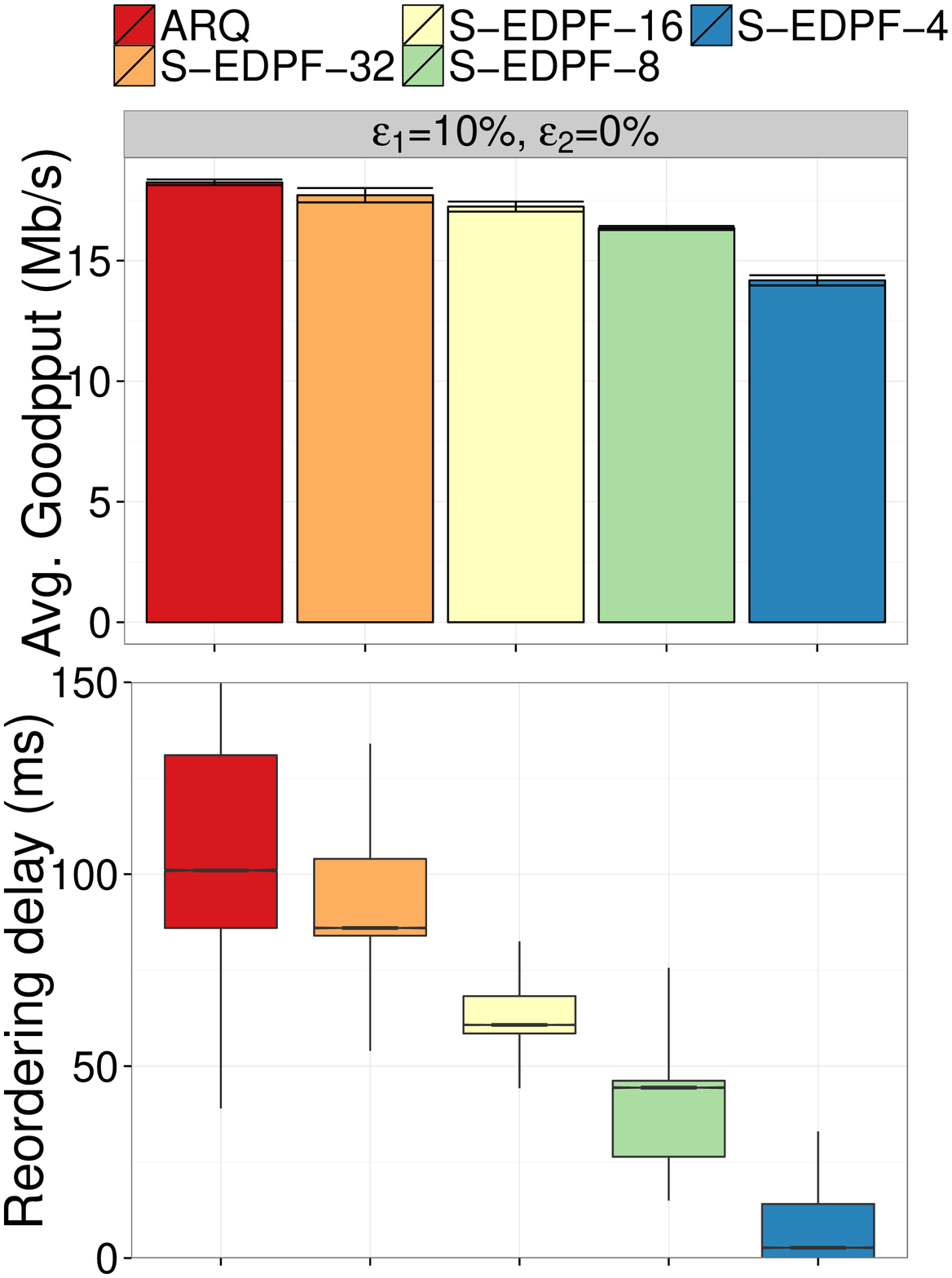}
	    \label{fig:exp:control_2:10}
        }
  \vspace{-2mm}
  \caption{Buffering delay and throughput upon fixed delays and controlled losses.}
  \vspace{-5mm}
  \label{fig:exp:control_2}
\end{figure}

 \subsubsection{Video Streaming}

%  The experimental results summarised above reveal considerable delay gains by S-EDPF. 
 In order to assess the impact of the throughput and delay performance gains observed above on real applications, we next evaluate the performance of S-EDPF with a popular video streaming service. We stream a video from  {\ttfamily Youtube} to a {\ttfamily vlc} video player and collect statistics of video frames displayed and dropped (e.g. due to excess buffering delays). We have selected an HD video (with resolution 1280x534), \emph{Star Wars VII teaser trailer 2},\footnote{\url{https://www.youtube.com/watch?v=wCc2v7izk8w} accessed on 26/5/2015.} encoded with H264/MPEG-4 AVC, with duration of 1:49 min, and a rate of $23.97$ fps. We configure the two access links with a bandwidth of 10~Mb/s, an RTT of 100~ms, and a variable random packet loss rate (for both data and ACKs) on one of the links; we don't drop packets on the other link. We also set {\ttfamily vlc} with $150$~ms of caching. Note that, although we emulate network conditions on the access links, we do not have control over the network between our proxy server and the {\ttfamily Youtube} server (i.e. Internet -- see Fig~\ref{fig:exp:testbed}). 
  
 \begin{figure}[t!]
   \vspace{-4mm}
\centering
\includegraphics[width=\columnwidth ]{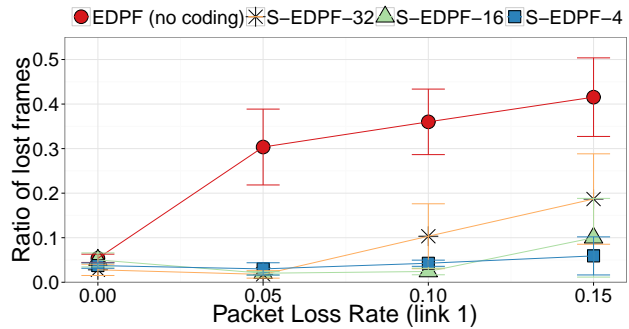}
  \vspace{-5mm}
\caption{{\ttfamily Youtube} video streaming.}
  \vspace{-5mm}
\label{fig:exp:video1}
\end{figure}
 
 We stream the video using EDPF (i.e., relying only on ARQ) and  S-EDPF-$\tau$ for different coding intervals $\tau$, and plot in Fig.~\ref{fig:exp:video1} the mean ratio of video frames that have not been displayed. We repeat each experiment 5 times.
 The results show a dramatic improvement on the streaming experience when using  S-EDPF. In particular, when the access link experiences losses, the video delivered with EDPF stutters and skips significant sections of the video on play back, as illustrated by the figure. In contrast, S-EDPF skips essentially no video frames and playout is consistently smooth.

\subsection{Complexity}

\begin{figure}[t!]
\centering
\vspace{-4mm}
\includegraphics[width=0.8\columnwidth ]{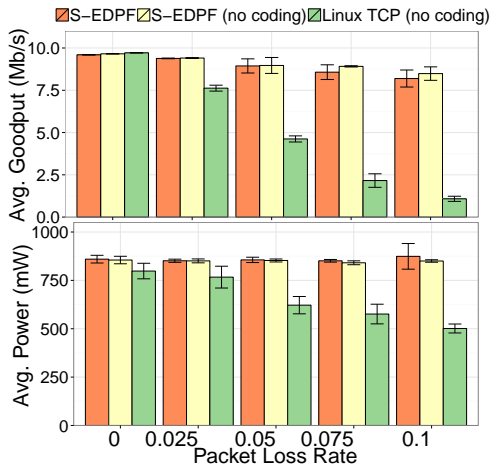}
\vspace{-6mm}
\caption{Energy consumption  and goodput of our prototype \emph{vs.} legacy sockets.}
\vspace{-5mm}
\label{fig:exp:energy}
\end{figure}

Given that our system adds an additional layer to the network stack (in user space), there is a risk of increasing the processing cost of the communication, and thus threatening the lifetime of battery-powered devices. We have performed a thorough profiling of our application,\footnote{Using tools like {\ttfamily valgrind/callgrind} and a power meter.} which highlights the encoder/decoder as the most costly element, as expected, because it uses CPU-intensive operations. In order to evaluate its complexity, we have installed our prototype in an Android LG G5 smartphone and we have measured the energy consumption of the device using a {\ttfamily Monsoon}\footnote{\url{https://www.msoon.com/LabEquipment/PowerMonitor/}} power meter. To simplify the setup, here we only use a local WiFi hotspot (i.e. a single path) with fixed rate (54~Mb/s), fixed CPU frequency, and backlight turned off.  
 We send TCP traffic from the device to the AP using {\ttfamily iperf} for 30~s, repeating each experiment 5 times and collecting the average throughput and power.  In the AP, we randomly drop packets at different rates and compare, in Fig.~\ref{fig:exp:energy},  the performance of S-EDPF, S-EDPF without our coding scheme (i.e. only relying on ARQ) and legacy TCP. For S-EDPF, we send as many extra coded packets as needed to compensate for losses. The first conclusion drawn from Fig.~\ref{fig:exp:energy} is that, when there are no losses, the throughput provided by both S-EDPF and TCP is the same, though at an extra energy cost of $\sim$30~mW. This represents a negligible $\sim$5\% of additional consumption (if we set the backlight with the highest level of brightness, this cost goes down to barely $\sim$1-2\%). When there are losses (not caused by congestion), the higher protection that S-EDPF's congestion control has versus legacy TCP's renders higher throughput performance.\footnote{See~\cite{ctcp} for a formal study of S-EDPF's congestion control.} Note that legacy TCP consumes considerably less energy simply because the bitrate is much lower. It is also worth noting that the additional encoding performed by S-EDPF (as compared to ``S-EDPF (no coding)'') does not seem to entail a significant energy burden.

 \subsection{``Into The Wild''} \label{sec:performance:wild}

Finally, we assess the performance of S-EDPF under real conditions. To this end, we use a laptop attached to a 3G/4G dongle and connect to different public WiFi hotspots around the city of Dublin, Ireland, as depicted in Fig.~\ref{fig:exp:testbed}. Our experimental campaign covers measurements in the campus of Trinity College Dublin, using a departmental WiFi network; a home environment, in a complex with many apartments; a public pub, during busy hours; St. Stephen's Green Mall in Dublin during a weekend day; and Dublin airport. 
  For each location, we download a large file for 30 s and repeat the measurement 5 times for each configuration: EDPF and S-EDPF-$\tau$ with different coding intervals $\tau$.  Fig.~\ref{fig:exp:wild} depicts the average goodput of the downloads and standard errors at the top subplot, and box and whiskers to measure packet delivery delay in the bottom subplot. The results illustrate how S-EDPF \emph{always} shows dramatic gains in delay performance, spanning from a mean delay reduction of 40\% in campus to a surprising 86\% in the airport when using $\tau=4$.
%   \footnote{The WiFi access provided in the airport was particularly variable and bursty. Even so, S-EDPF keeps delay low.} 
  It is also worthwhile mentioning that the throughput performance of EDPF in the ``Airport'' and ``Home'' tests does not improve over that when using only a single path (not shown here due to space constraints), a result that is consistent with that of \cite{chen2013measurement}'s. The reason is due to the \emph{excessive} buffering delay caused by both the losses and high variability in propagation delays and access rates. In contrast, even in these hostile environments, S-EDPF successfully combined the capacity of both links while keeping the mean packet in-order delivery delay low.

 \begin{figure*}[t!]
 \vspace{-5mm}
        \centering
             \subfloat[Campus.]{
	    \includegraphics[trim = 0mm 25mm 0mm 0mm, clip, width=0.18\linewidth]{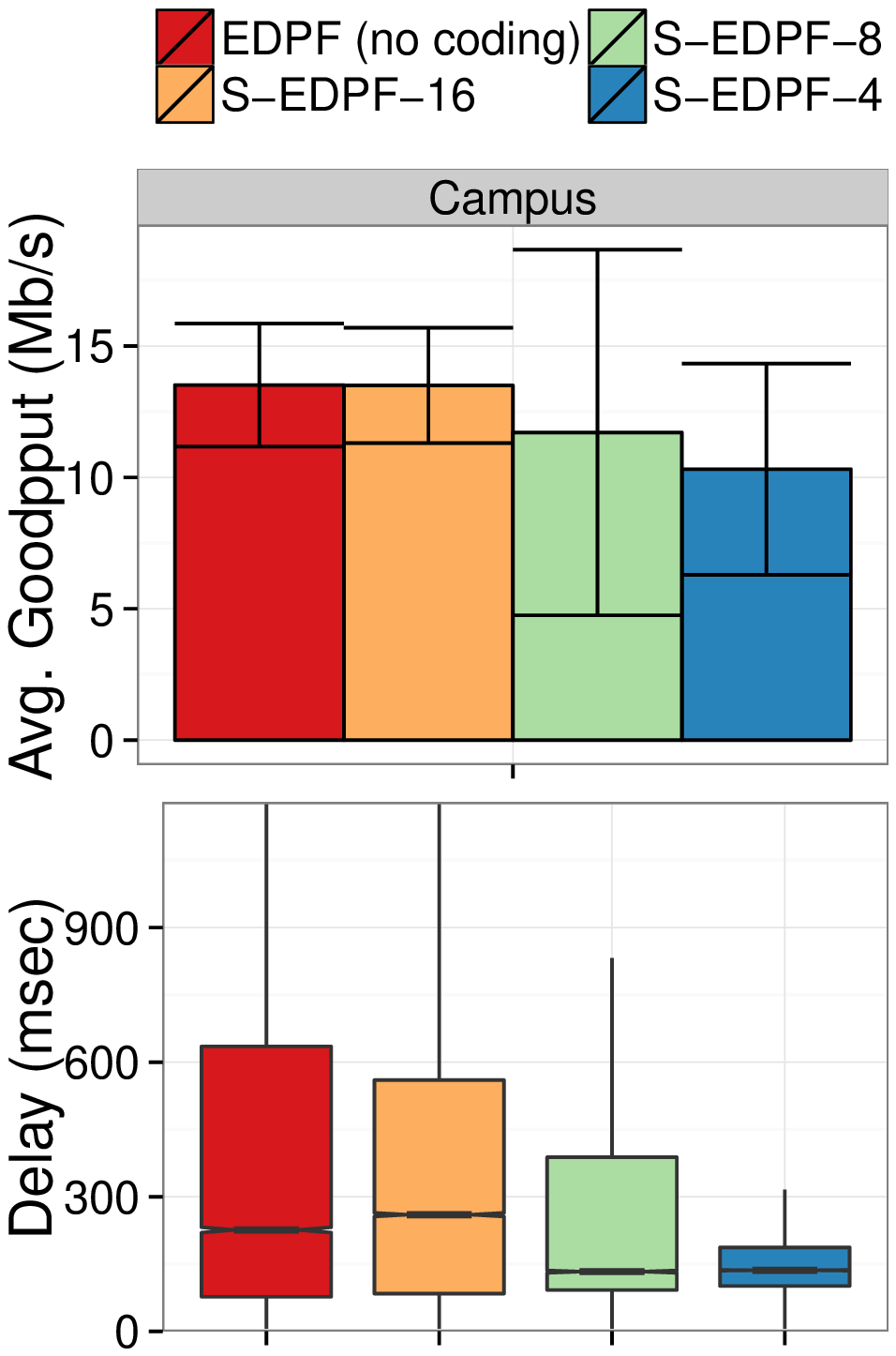}
	    \label{fig:exp:wild:campus}
        }%
        \subfloat[Home.]{
                \centering
	    \includegraphics[trim = 0mm 25mm 0mm 0mm, clip, width=0.18\linewidth]{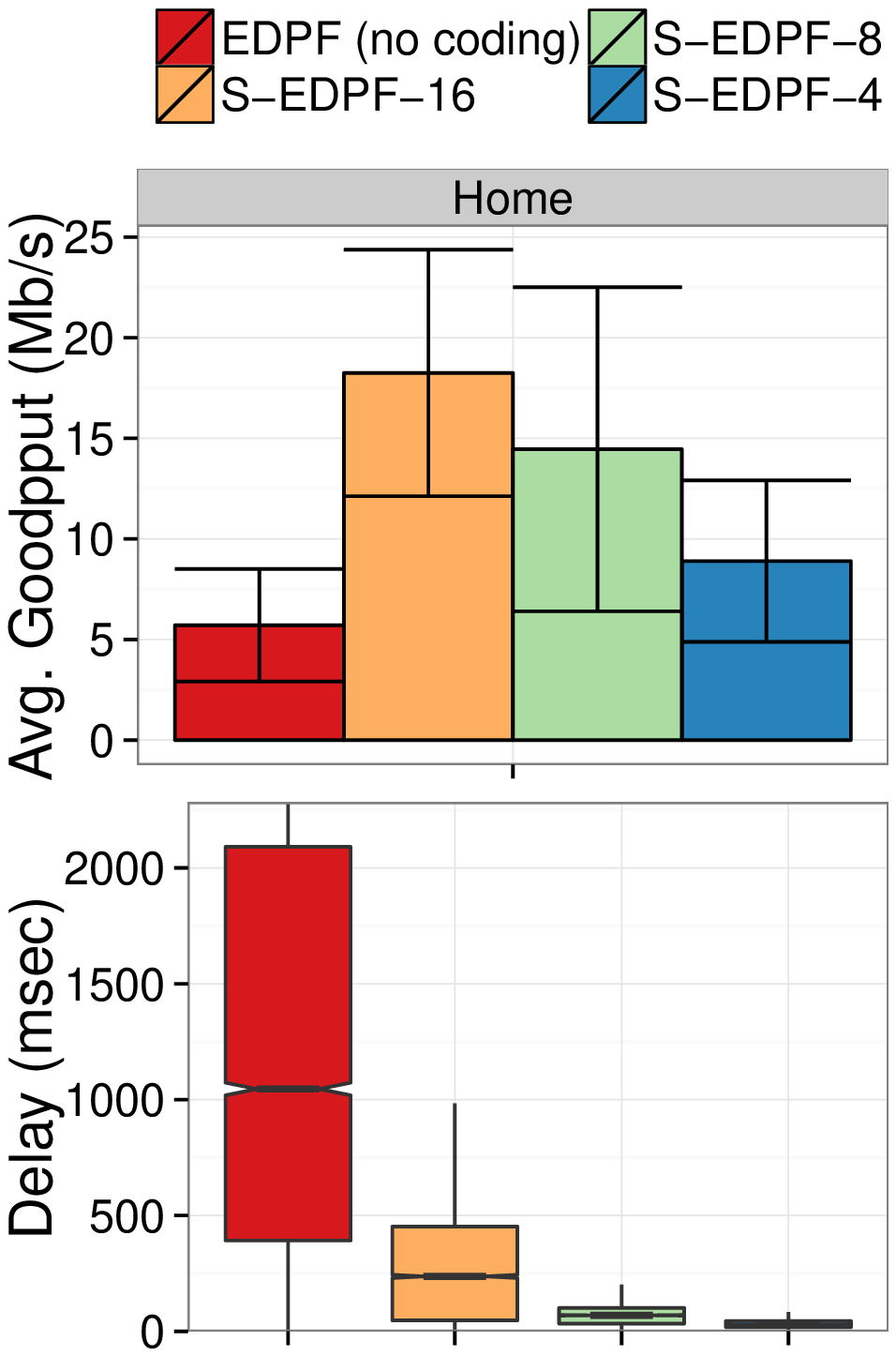}
	    \label{fig:exp:wild:home}
        }
        \subfloat[Pub.]{
                \centering
	    \includegraphics[trim = 0mm 25mm 0mm 0mm, clip, width=0.18\linewidth]{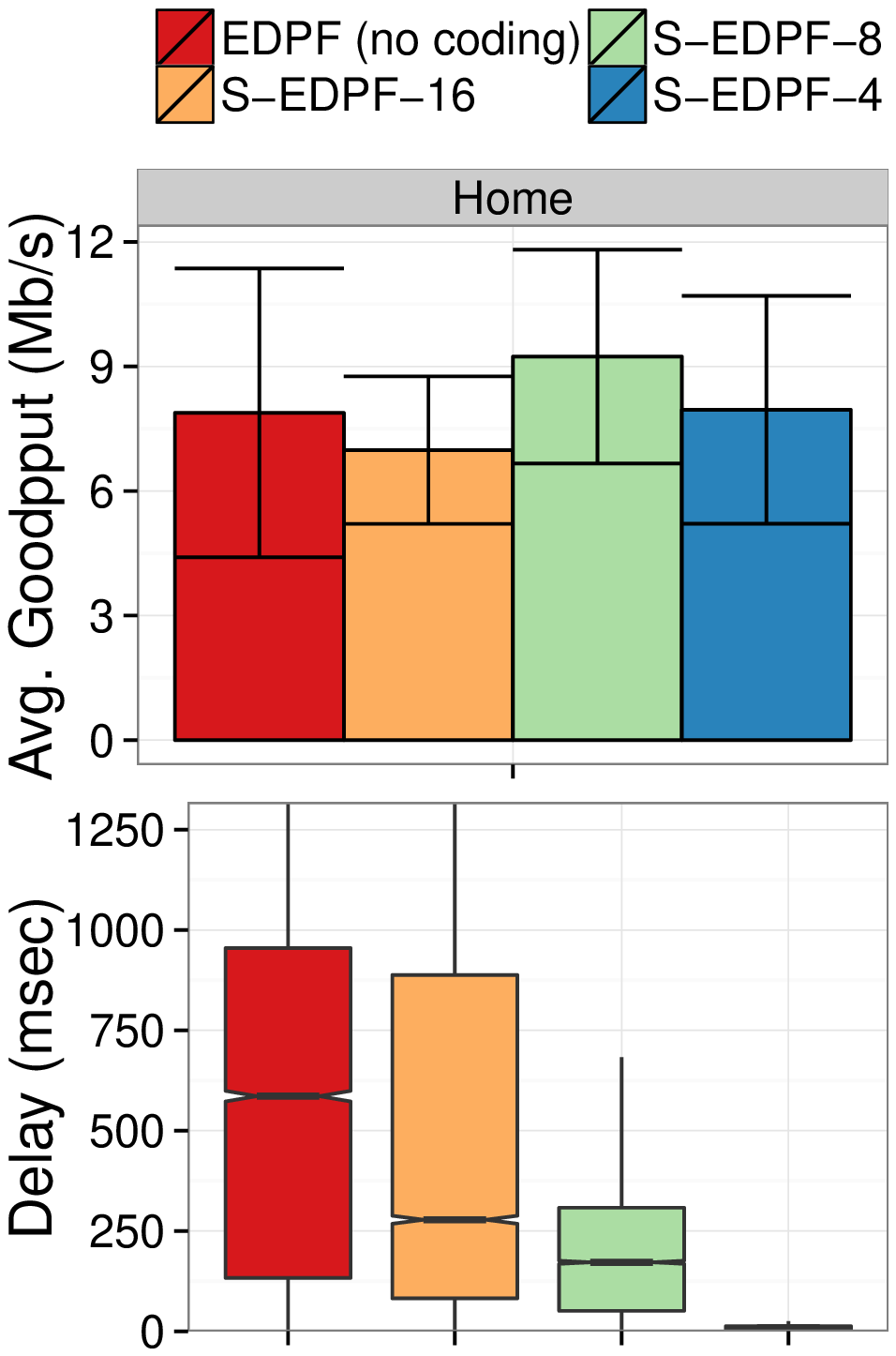}
	    \label{fig:exp:wild:pub}
        }
        \subfloat[Mall.]{
                \centering
	    \includegraphics[trim = 0mm 25mm 0mm 0mm, clip, width=0.18\linewidth]{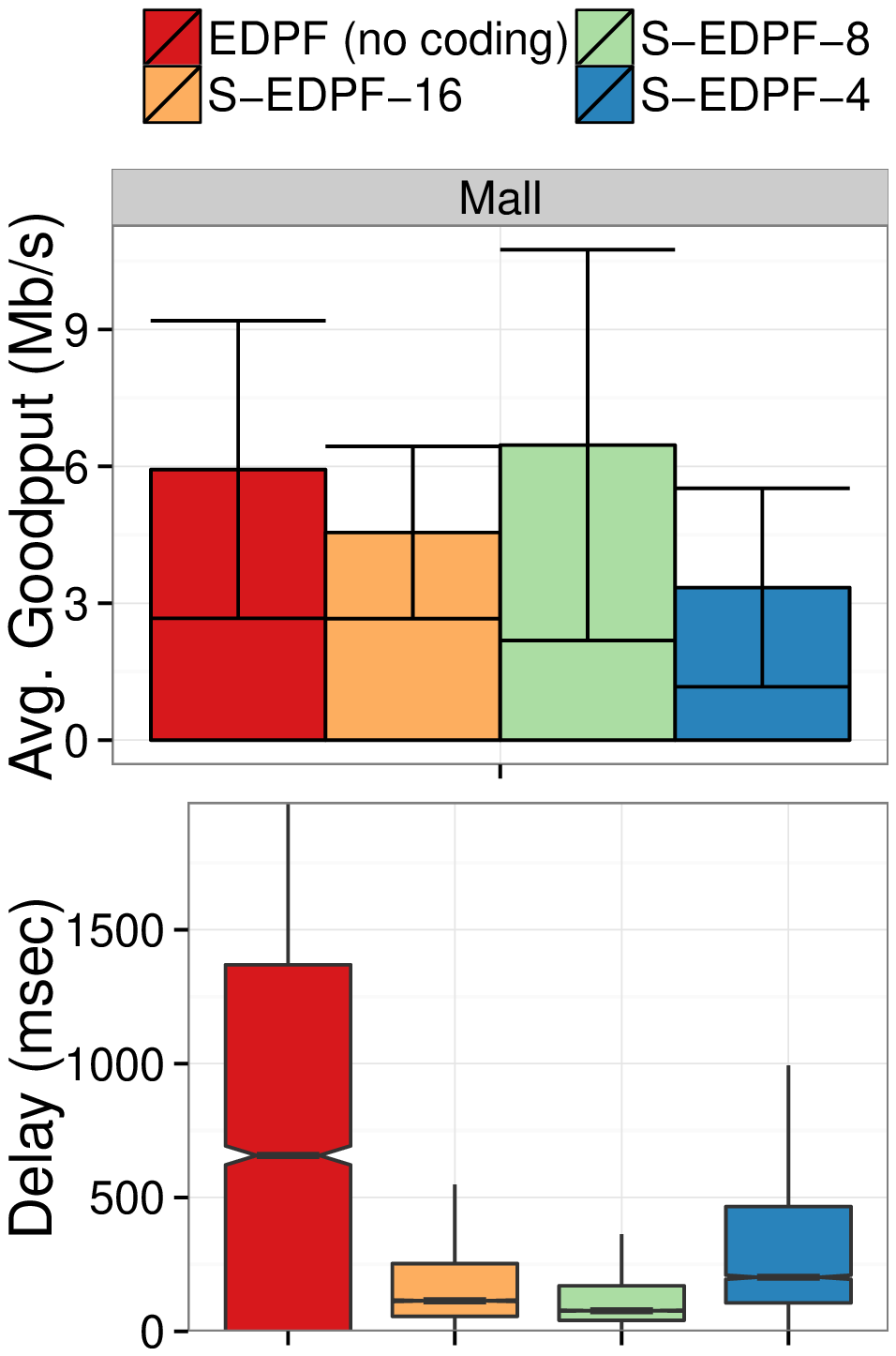}
	    \label{fig:exp:wild:mall}
        }
        \subfloat[Airport.]{
                \centering
	    \includegraphics[trim = 0mm 25mm 0mm 0mm, clip, width=0.18\linewidth]{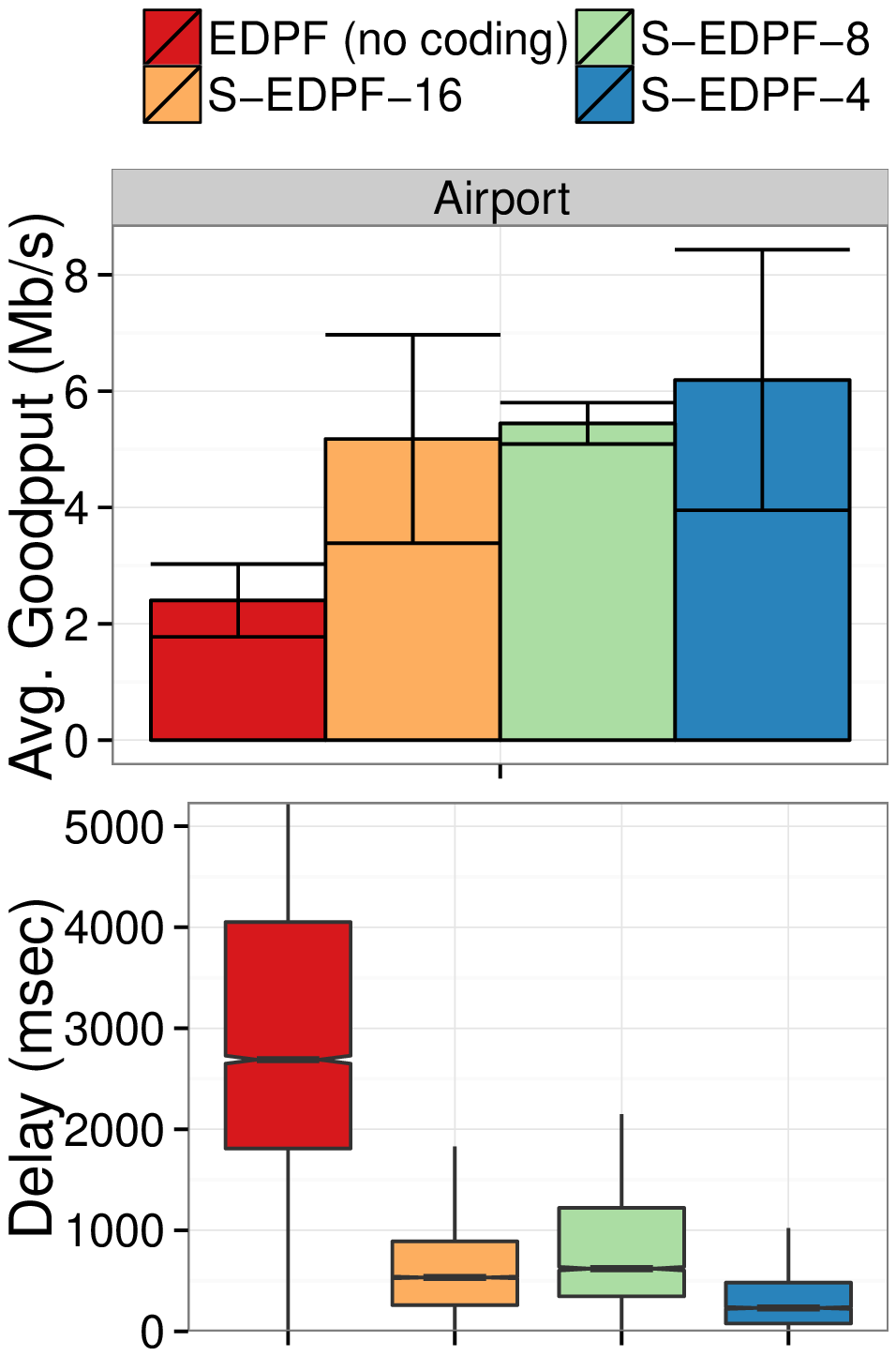}
	    \label{fig:exp:wild:airport}
        }        
        \vspace{-3mm}

  \caption{Experiments \emph{into the wild}.}
  \vspace{-2mm}
  \label{fig:exp:wild}
\end{figure*}

\section{Conclusions}

Multipath transport protocols are a promising technology to increase reliability and aggregate the capacity of multiple access providers. However, as we have  illustrated experimentally in this paper,  schedulers that do not consider transport delay variability, like LowRTT (MPTCP) or EDPF, suffer from a severe performance degradation in terms of delay that make them unsuitable for real-time applications. We have also shown that using ARQ to recover lost packets  further penalizes delay when round-trip times are large due to congestion or distance. To combat these issues, we propose in this paper a mechanism named Stochastic Earliest Delivery Path First (S-EDPF) that jointly schedules information and redundant (FEC) packets across the multiple network interfaces to ($i$) minimise the impact of packet reordering at the receiver when links are variable, and ($ii$) account for low-delay packet error recovery.

% implement a novel low-delay coding technique to address losses.

\section{Acknowledgements}
Work supported by SFI grants 11/PI/1177 and 13/RC/2077.

% \bibliographystyle{IEEEtran}
% \bibliography{multipath-coding}

\appendix

\section{Proofs}

\begin{proof}[Lemma \ref{lem:one}]
	
	Let $\alpha_{p,1},\alpha_{p,2},\dots,\alpha_{p,n}$ denote realisations of random variables $a_{p,1},a_{p,2},\dots,a_{p,n}$ i.e. the arrival times of packets sent in slot $n$ of path $p$.  
	
	Suppose packets on the same path are sent in ascending index order, i.e. for any $k\in\{1,s,\dots\}$ and $q_1,q_2\in K_{p,k}$ such that $q_1<q_2$ then $s_{p,q_1}<s_{p,q_2}$.  By Assumption \ref{assump:reorder}, it follows that $\alpha_{p,q_1}<\alpha_{p,q_2}$.  Let $\psi_{p,k} : = \max\{\alpha_{p,s_q}: q\in K_{p,k}\}$ and $\psi_k:=\max\{\psi_{1,k},\dots,\psi_{P,k}, \alpha_{p_k,s_k}\}$.  Then we have $\psi_{p,k} \le \psi_{p,k+1}$, $k\in\{1,2,\dots\}$, with equality only when $K_{p,k+1}=K_{p,k}$.  To see this, note that when $K_{p,k+1}=K_{p,k}$ equality is trivial, and when $K_{p,k+1}\ne K_{p,k}$ then $q^*_{k+1}:=\max\{q\in K_{p,k+1}\} > q^*_k:=\max\{q\in K_{p,k}\}$ due to the ascending order and so $\alpha_{p,q^*_{k+1}}>\alpha_{p,q^*_{k}}$ and thus $\psi_{p,k+1} > \psi_{p,k}$.  
	Hence, $\psi_k < \psi_{k+1}$, $k\in\{1,s,\dots\}$ i.e. the in-order delivery time is strictly increasing.   
	
	Suppose now that the transmission slots of two packets $q_1$ and $q_2$ with $q_1<q_2$ on path $p$ are swapped, so that the packets are now sent in non-ascending order $s_{q_1}>s_{q_2}$.   Let $\psi_{p,k}^{non} : = \max\{\alpha_{p,s_q}: q\in K_{p,k}\}$ and $\psi_k^{non}:=\max\{\psi_{1,k}^{non},\dots,\psi_{P,k}^{non},\alpha_{p_k,s_k}\}$. Then $\psi_{p,k}^{non} \le \psi_{p,k+1}^{non}$ but now for at least one $k\in\{1,2,\dots\}$, namely $k=q_2$, there will be equality when $K_{p,k+1}\ne K_{p,k}$ and so $\psi_k ^{non}= \psi_{k+1}^{non}$.  Further, for $k\ge q_2$ then $\psi_k ^{non}=\psi_k$ and also for $k< q_1$.  Hence, $\psi_k ^{non}>\psi_k$ for $k\in\{q_1,\dots,q_2-1\}$ and so the sum-delay is increased relative to sending packets in ascending order.      
	
	We can proceed by induction to show that swapping further packets cannot decrease the sum-delay below that when packets are sent in ascending order (there are two cases to consider, ($i$) when the further set of swapped packets is disjoint from those already swapped, in which case the above argument can be re-applied directly, and ($ii$) when the further set of swapped packets intersects with those already swapped, in which case we can recover an ascending packet order).
	
	Since the above holds for \emph{any} realisation $\alpha_{p,1},\dots,\alpha_{p,n}$, we are done.
\end{proof}

\begin{lemma}[Lyapunov Central Limit Theorem]\label{lem:clt}
	Suppose $\{X_1, X_2, \cdots \}$ is a sequence of independent random variables, each with finite expected value $\mu_i$ and variance $\sigma_i^2$. Let us define $ s_n^2 = \sum_{i=1}^n \sigma_i^2$.
	If for some $\delta > 0$ , the Lyapunov\textsc{\char13}s condition
	$$\lim_{n\to\infty} \frac{1}{s_{n}^{2+\delta}} \sum_{i=1}^{n} \mathbb{E}\big[\,|X_{i} - \mu_{i}|^{2+\delta}\,\big] = 0$$
	is satisfied, then a sum of $(X_i − \mu_i)/s_n$ converges in distribution to a standard normal random variable, as $n$ goes to infinity:
	$$\frac{1}{s_n} \sum_{i=1}^{n} (X_i - \mu_i) \ \xrightarrow{d}\ \mathcal{N}(0,\;1).$$
\end{lemma}

\begin{proof}[Theorem \ref{theorem:S-process}: S process]
The proof for 1)-4) goes as follows:
\begin{enumerate}
	\item We know that for the decoding process to go beyond $k$ we need at least $k$ erasures in the first $k$ frames. This means that there are cases with more than $k$ erasures in the first $k$ intervals that the decoding process stops before $k$ but for it to be greater than $k$ we should have at least $k$ erasures. Let $E_{k}$ denote the number of losses in $k$ frames. We have $ P(S>k | E_{k}\leq k )=0$. More formally:
	\begin{align}
		P(S > k ) & =   P(E_{kl}>k) P(S>k | E_{kl}>k ) \nonumber\\ 
		&\quad +P(E_{k}< k) P(S>k | E_{k}<k ) \nonumber\\
		& =   P(E_{k}>k) P(S>k | E_{k}>k )  \nonumber\\
		& <  P(E_{k}>k) \nonumber
	\end{align}
	
	$E_k$ is the summation of $\tau$ independent Bernoulli random variable with parameter $\epsilon_i$, in which $i$ depends on the path. By Lemma \ref{lem:clt}, we know the Lyapunov\textsc{\char13}s condition is satisfied and 
	$$\frac{1}{ \sqrt{ \sum_{p\in\mathcal{P}} kN_{\tau,p}\epsilon_p (1-\epsilon_p)}} (E_k - \sum_{p\in\mathcal{P}} kN_{\tau,p}\epsilon_p ) $$
	converges in distribution to $\mathcal{N}(0,\;1)$ for large values of $k$. $ P(E_{k}>k) $ goes to zero only if $ k $ is larger than $ \sum_{p\in\mathcal{P}} kN_{\tau,p}\epsilon_p$. This means 
	$$
	\sum_{p\in\mathcal{P}} N_{\tau,p}\epsilon_p <1.
	$$ 
	\item It is trivial to compute $P(S\!\!=k), \forall k\!\in\!\{0,1\}$ exactly. 
	
	Now, following the proof of Theorem 1.2 in \cite{karzand-2014}, we can compute $P(S=k),~k>1$  as:
	\begin{align}
	&P(S=k) = \label{eq:exact_S}\\
	&=\!\!\!\!\sum_{r=2}^{\min(k,l)}\!\!\!\!\!\! P(\mathcal{L}(\vv{\epsilon})\!=\!r)\frac{r\!\!-\!\!1}{k\!\!-\!\!1}P(\mathcal{L}(\left [\vv{\epsilon}\right ]_{\times (k-1)}\!\!=\!\!k\!\!-\!\!r)), \forall k>1\nonumber
	\end{align}
	where  $\vv\epsilon=\{\epsilon_{p_1}, \dots, \epsilon_{p_\tau}\}$, $\epsilon_{p_i}$ is the erasure probability in the subpath assigned to packet $i$, and  $\left [\vv{\epsilon}\right ]_{\times (k-1)}$ appends the $\vv{\epsilon}$ vector $k-1$ times.
	However, the computation of $P(\mathcal{L}(\vv{\epsilon})\!=\!r)$ requires a complex iterative algorithm~\cite{ehm-poisson_binomial}, which is not feasible to implement in S-EDPF. 
	We thus approximate this part of the S distribution with:
	\begin{align}\label{eq:distrib_S:tail_extended}
	&P(S=k) \approx P(S_1=k) = & \\
	&=\!\!\!\!\!\!\sum_{r=2}^{min(k,l)}\!\!\!\!\!\!P(\mathcal{B}(\bar\epsilon, \tau\!\!=\!\!r))\frac{r\!\!-\!\!1}{k\!\!-\!\!1}P(\mathcal{B}(\bar{\epsilon}, {\tau}(k\!\!-\!\!1)=k\!\!-\!\!r)) &\nonumber\\
	&= \frac{N_\tau}{k}    \bar\epsilon^{k} (1\!\!-\!\!\bar\epsilon)^{kN_\tau} {(k-1)\tau \choose k-1}, \forall k>1 \nonumber
	\end{align}
	where $\mathcal{B}$ is the binomial distribution with parameters $\bar\epsilon = \frac{\sum_{p\in\mathcal{P}} N_{\tau,p}\epsilon_p}{ N_\tau }$ and $N_\tau =   \sum\limits_{\substack{p\in\mathcal{P}}} N_{\tau,p} $.
	%    
	%   
	%  \item See \cite{ehm-poisson_binomial}\cite{roos-poisson_binomial}.
	\item 
% 	If we approximated the poisson-binomial distribution in eq.~(\ref{eq:exact_S}) with a binomial distribution of parameter $\bar\epsilon = \frac{\sum_{p\in\mathcal{P}} N_{\tau,p}\epsilon_p}{ N_\tau }$, we could compute an approximation of $\mathbb{E}[S_1]\approx \mathbb{E}[S]$ and $\mathbb{E}[S_1^2]\approx \mathbb{E}[S^2]$ directly from Theorem 1.3 in \cite{karzand-2014} as follows:
	Let us define $S_1$ as a random variable with the distribution of $S$ but with erasure probability $\bar\epsilon = \frac{\sum_{p\in\mathcal{P}} N_{\tau,p}\epsilon_p}{ N_\tau }$ for all subpaths. Following \cite{karzand-2014}, the first moment and second moments of $S_1$ are
	\begin{align}
	\mathbb{E}[S_1] &= \frac{(N_\tau-1) \bar\epsilon (1-\bar\epsilon)^{N_\tau-1}}{1-N_\tau \bar\epsilon} &\\
	\mathbb{E}[S_1^2] &= \mathbb{E}[S] + \frac{N_\tau(N_\tau-1) \bar\epsilon^2 (1- \bar\epsilon)^{N_\tau} }{(1-N_\tau \bar\epsilon)^3}
	\end{align}
% 	However, this approximation avoids the fact that we can compute $P(S=0)$ and $P(S=1)$ exactly. Thus, a much better approximation (see \S\ref{sec:code} for the rationale and validation) can be obtained as follows:	
	Now, $\mathbb{E}[S]$ can be approximated as follows:
	\begin{align}
	\mathbb{E}[S] \approx  P(S=1) + \sum_{k=2}^{\infty}kP(S_1=k).
	\end{align}
	Note that
% 	\begin{align}
	$
	\sum_{k=2}^{\infty}kP(S_1=k) = \mathbb{E}[S_1] - P(S_1 = 1)
	$,
% 	\end{align}
	Thus,
	\begin{align}
	\mathbb{E}[S]  &\approx P(S=1) + \mathbb{E}[S_1] - P(S_1=1)= & \nonumber\\
	&= P(S\!\!=\!\!1) + \frac{N_\tau(N_\tau \!\!-\!\!1)\bar\epsilon^2(1\!\!-\!\!\bar\epsilon)^{N_\tau\!-\!1}}{1\!\!-\!\!N_\tau\bar\epsilon},
	\end{align}
	given that $P(S_1=1) = (N_\tau-1)\bar\epsilon(1-\bar\epsilon)^{N_\tau-1}$.
	$\mathbb{E}[S^2]$ can be approximated similarly.
\end{enumerate}
\end{proof}

\begin{proof}[Theorem \ref{theorem:delayS}: In order delivery delay ]
	Let us assume that a transmission of a stream of $N_t$ packets is approximately a multiple of $\tau$ at time $t$. Upon this assumption, we have that, at time $t$, we have transmitted $\frac{N_t}{\tau} $ frames. Let us assume now that the decoding process consists of decoding periods of length $\{s_1, s_2, \dots , s_n, \dots\}$. Note that $\{S_1, S_2, \dots , S_n\}$ is sequence of  positive independent identically distributed random variables. Assume $S^+ = \min \{S, 1\},$ and define $J_n$ as follows:
	$J_n = \sum_{i=1}^n S_i^+,  n>0$; then the renewal interval $[J_n,J_{n+1}]$ is a decoding period. 
	Let $(X_t)_{t\geq0}$ count the decoding $\tau$-intervals that have occurred by time $t$, which is given by
	$$
	X_t :=  \sum^{\infty}_{n=1} \mathbb{I}_{\{J_n \leq t\}}=\sup \left\{\, n: J_n \leq t\, \right\}
	$$ 
	and is a renewal process ($\mathbb{I}$ is the indicator function).
	
	Let $W_1, W_2, \ldots$ be a sequence of $i.i.d.$ random variables denoting the sum of in order delivery delay in each decoding frame. We have two cases to consider.  Case (i): suppose the $j$'th period is an idle period.  Then $S_j=0$ and the information packets are delivered in-order with delay $\Delta_p$, where $\Delta_p$ is the transmission time of a packet sent in subpath $p$, here assumed as constant.  Case (ii): suppose the $j$'th period is a busy period and the information packet erasure that initiated the busy period started in the first slot $t_{i(j)}+1$.   Then the first information packet in each path $p$ is delayed by $S_j\Delta_pN_{\tau,p}$ slots, the second by $S_j \Delta_p N_{{\tau,p}}-1$ slots and so on.   The sum-delay over all of the information packets over all different paths in the busy period is therefore $\sum_{p\in\mathcal{P}} (\sum_{k=1}^{S_j\Delta_p N_{\tau,p}}k - \sum_{k=0}^{S_j-1}k\Delta_p N_{\tau,p}) < \sum_{p\in\mathcal{P}}\frac{S_j^2\Delta_p N_{\tau,p}(\Delta_p N_{\tau,p}-1)}{2}$.  
	The random variable $Y_t = \sum_{i=1}^{X_t}W_i $ is a renewal-reward process and its expectation is the sum of in order delivery delay over the time-span of $t$. 
	Based on the elementary renewal theorem for renewal-reward processes \cite{gallager13}, we have:
	\begin{equation}\label{eqn:reward}
	\lim_{t \to \infty} \frac{1}{t} \mathbb{E}[Y_t]   = \frac{\mathbb{E}(W_1)}{\mathbb{E}[S_1^+]}.
	\end{equation}
	Based on the construction of $W_i$, we have
% 	\begin{align}
	$\mathbb{E}[W_1]  = \sum_{i=1}^{\infty} \mathbb{E}[W_1|S_1^+= i] \Pr(S_1^+=i) = \sum_{i=0}^{\infty} \mathbb{E}[W_1|S_1= i] \Pr(S_1=i) = \mathbb{E}[S^2] \sum_{p\in\mathcal{P}}  \frac{\Delta_p N_{\tau,p}(\Delta_p N_{\tau,p}-1)}{2}$.
% 	\end{align}
% 	\tbd{I think here we are missing the case where $S=0$; note that now the delay is not zero in this case but $\Delta_p$ Mohammad: We are concerning the buffering delay in the receiver. }.
	We also have that 
	$
	\mathbb{E}[S_1^+] = \mathbb{E}[S] + P(S=0)
	$, and that 
	$
	\lim_{t \to \infty} \frac{1}{t} \mathbb{E}[Y_t] \leq  \frac{\mathbb{E}[S^2]}{ \mathbb{E}[S] + P(S=0)} \sum_{p\in\mathcal{P}}  \frac{\Delta_p N_{\tau,p}(\Delta_p N_{\tau,p}-1)}{2}. 
	$
	Since we assumed that $N_t = t {\tau}$, we have:
	$$
	\lim_{N_t \to \infty} \frac{1}{N_t} {E}[Y_t] \leq \frac{ \sum_{p\in\mathcal{P}}  \Delta_p N_{\tau,p}(\Delta_p N_{\tau,p}-1)}{ 2{\tau}(\mathbb{E}[S] + P(S=0))}\mathbb{E}[S^2]
	$$
\end{proof}

\begin{proof}[Theorem \ref{theorem:coded-path}: Path choice for coded packets]
	Suppose we have a system with $\mathcal{P}=\{1,\dots,p\}$ subpaths with erasure probabilities $\epsilon_1\ge\epsilon_2\ge\dots\ge\epsilon_p$.  Let us first evaluate two independent $S$ processes, $S_1$ and $S_2$, with $p_c=1$ and $p_c=p_2$ respectively, i.e., we schedule our coded packets on path $1$ in the first case, and $p_2$ (any other) second. We know from Theorem~\ref{theorem:S-process} that the path selected for the coded packets \emph{does not} affect the decoding process (delay) when $S>1$, i.e., 
	\begin{align}\label{eq:coded-path:equal}
	P(S_1=k) = P(S_2=k),~\forall k>1.
	\end{align}
	
	When $\sum_{p\in\mathcal{P}} N_{\tau,p} \epsilon_p < 1$ (i.e. when we operate below the capacity limit), the above yields:
	\begin{align}\label{eq:coded-path:balance}
	P(S_1=0)\!+\!P(S_1=1)\!=\!P(S_2=0)\!+\!P(S_2=1)
	\end{align}
	given that $\sum_k P(S=k) = 1$.
	
	Note that the decoding delay is equal to zero slots with probability $P(S=0)$ (because it is the in-order delivery state), and non-zero otherwise (packets are being buffered until all losses are recovered). This means that, given equations (\ref{eq:coded-path:equal}) and (\ref{eq:coded-path:balance}), \emph{the $S_i$ process that achieves lower delay is the one that maximises $P(S=0)$}. 
	Let us then compare $P(S_1=0)$ and ${P(S_2=0)}$:
	\begin{align}
	\frac{P(S_1=0)}{P(S_2=0)} &= \frac{ (1-\epsilon_{1})^{N_{\tau,1}-1}\prod_{i \ne 1} (1-\epsilon_i)^{N_{\tau,i}} }{ (1-\epsilon_{p_2})^{N_{\tau,p_2}-1}\prod_{i \ne p_2} (1-\epsilon_i)^{N_{\tau,i}} } = \nonumber\\
% 	&= \frac{ (1-\epsilon_{1})^{N_{\tau,1}-1}(1-\epsilon_{p_2})^{N_{\tau,p_2}} } { (1-\epsilon_{p_2})^{N_{\tau,p_2}-1}(1-\epsilon_{1})^{N_{\tau,1}} } = \nonumber \\
 &=\frac{(1-\epsilon_{p_2})}{(1-\epsilon_{1})} \ge 1 \label{eq:coded-path-n2}
	\end{align}
	because $1 \ge \epsilon_1 \ge \epsilon_{p_2} \ge 0$ . Thus
	$
	P(S_1=0) \ge P(S_2=0)
	$,
	i.e., scheduling coded packets on path 1 (the one with highest loss probability) minimises the decoding delay, and this is valid for any $S_i$ process different than $S_1$.
\end{proof}

\end{document}